\documentclass{amsart}
\usepackage{amsmath,amssymb} 
\usepackage{amsbsy}
\usepackage{latexsym}
\usepackage{mathrsfs}
\usepackage{dsfont}
\usepackage{color}
\definecolor{red}{rgb}{.7,0,0}

\newcommand{\N}{\mathbb N}
\newcommand{\R}{\mathbb{R}}
\newcommand{\C}{\mathbb{C}}
\newcommand{\Z}{\mathbb{Z}}
\newcommand{\E}{\mathbb{E}\,}

\newcommand{\z}{\zeta}
\renewcommand{\a}{\alpha}
\renewcommand{\b}{\beta}
\renewcommand{\d}{\delta}
\newcommand{\s}{\sigma}
\newcommand{\e}{\varepsilon}

\newcommand{\LV}{\mathrm{LV}} 
\newcommand{\Oh}{\mathcal{O}}

\newcommand{\ind}[1]{\mathds{1}_{#1}}
\newcommand{\ddd}{\cdot\ldots\cdot}

\def\ND{\mathcal{N}}

\newtheorem{thm}{Theorem}
\newtheorem{lemma}[thm]{Lemma}
\newtheorem{cor}[thm]{Corollary}
\newtheorem{prop}[thm]{Proposition}

\newtheorem{defi}[thm]{Definition}

\newtheorem{conj}{Conjecture}

\theoremstyle{remark} 
\newtheorem{remark}[thm]{Remark}
\newtheorem{example}[thm]{Example}

\numberwithin{equation}{section}
\numberwithin{thm}{section}
\title{The real tau-conjecture is true on average} 
\date{\today}

\author{Ir\'en\'ee Briquel}
\address{Universit\'e de Cergy-Pontoise}
\email{irenee.briquel@u-cergy.fr} %gmail.com}

\author{Peter B\"urgisser}
\thanks{PB partially supported by DFG grant BU 1371/2-2 
and by the ERC under the European's Horizon~2020 research and innovation programme (grant agreement no.~787840).}

\address{Institute of Mathematics, Technische Universit\"at Berlin}
\email{pbuerg@math.tu-berlin.de}

\begin{document}

\keywords{zeros of random polynomials, Descartes rule, sparsity, depth four arithmetic circuits, tau-conjecture, complexity theory}

\subjclass[2010]{26C10, 60G60, 68Q17}

\maketitle

\begin{abstract}
Koiran's real $\tau$-conjecture claims that 
the number of real zeros of a structured polynomial given 
as a sum of $m$ products of $k$ real sparse polynomials, 
each with at most $t$ monomials,  
is bounded by a polynomial in $mkt$. 
This conjecture has a major consequence in complexity theory 
since it would lead to superpolynomial lower bounds for the arithmetic circuit size 
of the permanent. We confirm the conjecture in a probabilistic sense by 
proving that if the coefficients involved in the description of $f$
are independent standard Gaussian random variables, then 
the expected number of real zeros of $f$ is $\Oh(mk^2t)$.
%which is linear in the number of parameters. 
\end{abstract}

%\tableofcontents

\section{Introduction}

We study the number of real zeros of real univariate polynomials. 
A polynomial~$f$ is called $t$-sparse if it has at most~$t$ monomials. 
Descartes rule states that a $t$-sparse polynomial~$f$ has at most $t-1$ 
positive real zeros, no matter what is the degree of $f$. 
Therefore, a product $f_1\cdots f_k$ of $k$ many $t$-sparse polynomials $f_j$ 
can have at most $k(t-1)$ positive real zeros. 
What can we say about the number of zeros of a sum of $m$ many products? 
So we consider real univariate polynomials $F$ of the following structure
\begin{equation}\label{eq:D4-poly}
 F = \sum_{i=1}^m \prod_{j=1}^{k_i} f_{ij}
\end{equation}
where all $f_{ij}$ are $t$-sparse. 
%{\tt Could exchange $m$ and $k$ for consistency with Koiran's notation.} 
In other words, $F$ is given by a depth four arithmetic circuit with the structure
$\Sigma\Pi\Sigma\Pi$, where the parameters $m$, $k:=\max_i k_i$, and $t$ bound the fan-in at the different levels 
except at the lowest (since we don't require a bound on the degrees of the $f_{ij}$). 

The following conjecture was put forward by Koiran~\cite{koiran:11}.

\begin{conj}[Real $\tau$-conjecture]\label{conj:RealTau}
The number of real zeros of a polynomial $F$ of the form \eqref{eq:D4-poly} 
is bounded by a polynomial in $m$, $k$, and $t$. 
\end{conj}

Koiran~\cite{koiran:11} proved that the real $\tau$-conjecture implies a 
major conjecture in complexity theory, namely the separation of complexity classes  
$\mathrm{VP}^0 \ne  \mathrm{VNP}^0$ over $\C$.
In  Tavenas' PhD thesis~\cite{tavenas-thesis} 
it is shown that the real $\tau$-conjecture also implies 
that $\mathrm{VP}\ne  \mathrm{VNP}$ over $\C$. 
Tavenas also shows that a seemingly much weaker upper bound on the number of real zeros of $F$ 
is sufficient to deduce $\mathrm{VP}\ne  \mathrm{VNP}$: in fact, 
an upper bound polynomial in $m,t, 2^{\max_i k_i}$ is sufficient~\cite[\S 2.1, Cor.~3.23]{tavenas-thesis}.   
In other words, the real $\tau$-conjecture implies that the permanent of $n$ by $n$ matrices 
requires arithmetic circuits of superpolynomial size. 
For known upper bounds on the number of real zeros of polynomials of
the form~$F$, we refer to~\cite{koir-port-tav:15} and the references
given there. 

The motivation behind Conjecture~\ref{conj:RealTau} is Shub and Smale's $\tau$-conjecture~\cite{shub-smale:95} 
asserting that the number of integer zeros of a polynomial computed by an arithmetic circuit 
is polynomially bounded by the size of the circuit. 
If true, it gives a superpolynomial lower bound on the circuit complexity of the permanent polynomial~\cite{buerg:09}.
Moreover, it also entails the separation $\mathrm{P}_\C\ne\mathrm{NP}_\C$ 
in the Blum-Shub-Smale model~\cite{shub-smale:95,BCSS:98}. 
One drawback of the $\tau$-conjecture is that, by referring to integer zeros, it leads to number theory, 
which is notorious for its hard problems.  
The $\tau$-conjecture is false when we replace ``integer zeros' by ``real zeros''. 
Koiran's observation is that when restricting to depth four circuits, the conjecture may be true 
and we can still derive lower bounds for general circuits. 
%Koiran's proof relies on a technique in~\cite{buerg:09}. 
We refer to Hrubes~\cite{hrubes:13} for statements  equivalent to the real $\tau$-conjecture 
that are related to complex zero counting. 
A $\tau$-conjecture for the Newton polygons of bivariate polynomials,
having the same strong complexity theoretic implications,  
has been formulated by Koiran et al.\ in~\cite{kptt:15}. 
Hrubes~\cite{hrubes:19} recently showed that the real $\tau$-conjecture 
implies this conjecture on Newton polytopes.

In this work, we prove that the real $\tau$-conjecture is true for random polynomials. 
More specifically, let $k_1,\ldots, k_m$ and $t$ be positive integers 
and for $1\le i\le m$ and $1 \le j \le k_i$ we fix supports $S_{ij}\subseteq\N$ with 
$|S_{ij}| \le t$ for the $t$-sparse polynomials $f_{ij}$. 
We choose the coefficients $u_{ijs}$ of the polynomials 
$$
 f_{ij}(x) = \sum_{s\in S_{ij}} u_{ijs} x^s
$$
as independent standard Gaussian random variables. 
The resulting $F$ given by~\eqref{eq:D4-poly}
is a structured random polynomial and we investigate the random variable 
defined as the number of real zeros of $F$. 

Our main result states that the expectation of the 
number of real zeros of $F$ is polynomially 
bounded in $m$, $k:=\max_i k_i$, and $t$. 
In fact, we get an at most quadratic bound 
in the number of parameters!

\begin{thm}\label{th:main}
The expectation of the number of real zeros of a polynomial $F$ 
of the form \eqref{eq:D4-poly} is bounded as 
$\Oh( m k^2 t)$
if the coefficient $u_{ijs}$ are independent and standard Gaussian. 
Thus the real $\tau$-conjecture is true on average.
\end{thm}
%(with a linear upper bound). 

Our result can be interpreted in two ways: on the one hand, 
it supports the real $\tau$-conjecture since we show it is true on average; 
on the other hand it says that for finding a counterexample to the 
real $\tau$-conjecture, it is not sufficient to look at generic examples.

We don't think the assumption of Gaussian distributions is relevant. 
In fact, we have a partial result confirming this (Theorem~\ref{th:SP1}).
If we assume the coefficients $u_{ijs}$ are independent random variables whose distribution have 
densities satisfying some mild assumptions, then 
the expected number of real zeros of $F$ in $[0,1]$ is bounded 
by a polynomial in $k_1+\ldots + k_m$ and $t$, provided $0\in S_{ij}$ for all $i,j$.
The latter condition means that all the $f_{ij}$ almost surely have a nonzero constant coefficient. 

The main proof technique is the Rice formula from the theory of random fields, which has 
to be analyzed very carefully in order to achieve the good upper bounds. 
(In fact, we rely on a ``Rice inequality'', which requires less assumptions.) 
An interesting intermediate step of the proof is to express the 
expected number of real zeros of the random structured~$F$ from~\eqref{eq:D4-poly} 
in terms of the expected number of real zeros of random linear combinations 
$R(x) = \sum_{i=1}^m u_i q_i(x) x^{d_i}$ of certain weight functions $q_i(x) x^{d_i}$. 
The deterministic functions $q_i(x)$ are obtained by multiplying and dividing 
sparse sums of squares in a way reflecting the build-up of the arithmetic circuit forming~$F$;  
see~\eqref{def:qi}. 
The randomness comes from independent coefficients~$u_i$, 
whose distribution is the one of a product of $k_i$ standard Gaussians. 

It would be interesting to strengthen our result by concentration statements, showing that 
it is very unlikely that a random $F$ of the above structure can have many real zeros.

\subsection*{Outline of paper} 
Section~\ref{se:prelim} provides hands-on information on how to 
deal with conditional expectations, which is mainly basic calculus. 
In Section~\ref{se:Rice} we outline the idea of the Rice formula and state 
a weak version of it (Theorem~\ref{th:rice}), which requires only few technical 
assumptions. 
In Section~\ref{se:BCE} we prepare the ground by proving general estimates 
on conditional expectations of random linear combinations. 
Section~\ref{se:zerosLC} develops general results of independent interest 
on the  expected number of real zeros of random linear combinations 
$\sum_{i=1}^m w_i(x) u_i$ of weight functions~$w_i$, for independent 
random coefficients $u_i$ having densities satisfying some mild assumptions. 
We upper bound this in terms of quantities $\LV(w_i)$, for which we coined 
the name {\em logarithmic variations}, 
and which are crucial for achieving good estimations 
(see Definition~\ref{def:def-LV}). 
Finally, combining everything, we provide the proof of the main results in 
Section~\ref{se:GSPSP}.

\subsection*{Acknowledgments} 
We thank Pascal Koiran and Mario Kummer for helpful discussions. 
Peter B\"urgisser is grateful to the late Mario Wschebor for
introducing him into the wonders of the Rice formula. 
We thank the anonymous referees whose comments led to an improved 
presentation.

%%%
\section{Preliminaries}\label{se:prelim}

We provide some background on conditional expectations in a general continuous setting, 
relying on some results from calculus related to the coarea formula. 
Then we discuss some specific properties pertaining to the distribution of products of 
Gaussian random variables.

\subsection{Conditional expectations}\label{se:cond-exp}

We fix a smooth function $f\colon\R^N\to\R$ 
with the property that $\{u\in\R^N : \nabla f(u) = 0 \}$ has measure zero. 
In most of our applications, $f$ will be a nonconstant polynomial function, which satisfies this property. 
By Sard's theorem, almost all $a\in\R$ are regular values of $f$.  
For those $a$, the fiber $f^{-1}(a)$ is a smooth hypersurface in $\R^N$. 

Suppose we are further given a probability distribution  on $\R^N$ with the density $\rho$. 
To analyze its pushforward measure with respect to $f$, we define 
for a regular value $a\in \R$
\begin{equation}\label{eq:pushf_dens}
 \rho_f(a) := \int_{f^{-1}(a)} \frac{\rho}{\|\nabla f\|}\, df^{-1}(a) \ \in\  [0,\infty];  
\end{equation}
here $df^{-1}(a)$ denotes the volume element of the hypersurface $f^{-1}(a)$. 
The coarea formula is a crucial tool going back to Federer~\cite{federer:59}, 
see \cite[Thm.~III.5.2, p.~138]{chavel:06} for a comprehensive account. 
We only need its smooth version~\cite[p.~159]{chavel:06};  
see also~\cite[Appendix]{howard:93} for a short and self-contained proof. 
The smooth coarea formula implies that 
$\rho_f$ defined in \eqref{eq:pushf_dens} is a probability density on $\R$, namely 
the density of the random variable $f(a)$. 
More precisely, $\rho_f$ is the 
{\em pushforward measure} with respect to $f$ of the measure on $\R^N$ with density $\rho$.

Let us point out the following simple rule, which we will use all the time: for $\lambda\in\R^*$
\begin{equation}\label{eq:TR}
 \rho_{\lambda f}(\lambda a) = \frac{1}{|\lambda|} \rho_f(a) .
\end{equation}

We view now $u\in\R^N$ as a random variable with the density~$\rho$. 
Let $a\in\R$ be a regular value of~$f$ such that $\rho_f(a) >0$. 
We want to define a conditional probability measure on the hypersurface $H:=f^{-1}(a)$ 
that captures the idea that we constrain $u$ to lie in $H$. 
We do this by defining the {\em conditional density} for $u\in H$ as 
\begin{equation*}%\label{eq:cond-density}
\rho_H (u) := \frac{1}{\rho_f(a)} \frac{\rho(u)}{\|\nabla f(u)\|} .
\end{equation*}
Note that  we indeed have $\int_H \rho_H\, dH =1$ by construction, 
where $dH$ denotes the volume measure of $H$. 
(As a warning, let us point out that in general, $\rho_H$~does not only depend 
on $H$, but also on the representation of $H$ by the function~$f$.)
%If $Z\colon\R^N\to [0,\infty]$ is a nonnegative random variable, %which is integrable with respect to $\rho$, 
Using the conditional density, 
we can define the {\em conditional expectation} 
$$
  \E(Z \mid f = a) := \int_H Z \rho_H dH \ \in\ [0,\infty]
$$
of a nonnegative measurable function $Z\colon\R^N\to [0,\infty]$.
(This quantity is only defined for regular values $a$ such that $\rho_f(a) >0$.)
In our application, we will always use the following equivalent formula
\begin{equation}\label{eq:Erho}
 \E(Z\mid f = a) \, \rho_f(a) = \int_H Z\,\frac{\rho}{\|\nabla f\|}\, dH , 
\end{equation}
which is valid for all regular values $a$ of $f$,  
when interpreting the left-hand side as $0$ if $\rho_f(a)=0$. 
Thus by Sard's theorem, the equation makes sense for almost all $a\in\R$

After defining all these notions, we summarize our discussion by stating the following 
important fact, which is an immediate consequence of the smooth coarea formula
(cf.~\cite[p.~159]{chavel:06} or~\cite[Appendix]{howard:93}). 

\begin{prop}\label{pro:coarea}
Let $f\colon\R^N\to\R$ be a smooth function such that 
$\{u\in\R^N : \nabla f(u) = 0 \}$ has measure zero. 
Moreover, let $\rho$ be a probability density on $\R^N$ and 
$Z\colon\R^N\to [0,\infty]$ be measurable. Then we have 
\begin{equation*}\label{eq:coarea}
\E (Z) = \int_\R \E(Z\mid f = a) \, \rho_f(a) \, da .
\end{equation*}
\end{prop}

%The smooth coarea formula implies that 
%\begin{equation}\label{eq:coarea}
%\E (Z) = \int_\R \E(Z\mid f = a) \, \rho_f(a) \, da .
%\end{equation}
%(Recall that we assume that $\nabla f (x)=0$ occurs with probability zero.) 

We next discuss how to compute the right-hand side in concrete situations. 
As a first step, we express the volume element of the hypersurface $H$ in local coordinates.
If $\partial_{u_1}f \ne 0$, then by the implicit function theorem, 
we can locally express $u_1$ as a function of $u_2,\ldots,u_N$. 
The following lemma is well known. For the understanding of the following, 
it is helpful provide the proof.
%Then we can express the volume element of $H$ in the local coordinates $u_2,\ldots,u_N$ as follows.

\begin{lemma}\label{le:dH_loc_coord}
We have 
$$
 dH = \frac{\|\nabla f\|}{|\partial_{u_1} f|} \, du_2\cdots du_ N .
$$
\end{lemma}
 
\begin{proof}
Generally, if we parametrize $H$ by $u=\psi(t_1,\ldots,t_{N-1})$, using local 
coordinates $t_1,\ldots t_{N-1}$, it is well known that the volume element of $H$
is given by $dH = \sqrt{\det((D\psi)^T D\psi)}\,dt_1\ldots t_{N-1}$. 
In our situation, we locally write $u_1=h(u_2,\ldots,u_N)$ 
and use the parametrization 
$\psi(u_2,\ldots,u_N) := (h(u_2,\ldots,u_N),u_2,\ldots,u_N)$ of $H$.  
%using the local coordinates $t_i=u_{i+1}$ 
A straightforward calculation shows 
$(D\psi)^T D\psi = I + \nabla h (\nabla h)^T$.
Moreover, 
$\det( I + \nabla h (\nabla h)^T) = 1 + \|\nabla h \|^2$. 
(In order to see this, use the orthogonal matrix $S\in O(N)$ such that 
$S\nabla h = (0,\ldots,0,\|\nabla h\|)$.) 
Hence the volume element of $H$ %in the given coordinates 
satisfies 
\begin{equation*}%\label{eq:vol-el}
 dH = \sqrt{1 + \|\nabla h\|^2} \, du_2\cdots du_{N} . 
\end{equation*}
By implicit differention we get 
$\partial_{u_i}h = -\partial_{u_i}f/\partial_{u_1}f$. 
Hence,
$$
 1 + \|\nabla h\|^2 = \frac{\|\nabla f\|^2}{(\partial_{u_1}f)^2} ,
$$
and the assertion follows.
\end{proof}

Assume now that $H$ is parametrized when $(u_2,\ldots,u_N)$ runs over 
(an open dense subset of) $\R^{N-1}$.
Then, due to Lemma~\ref{le:dH_loc_coord}, we can express the pushforward density $\rho_f$ as follows:
\begin{equation}\label{eq:rholin} 
 \rho_f(a) = \int_{\R^{N-1}} \frac{\rho}{|\partial_{u_1}f|}\, du_2\cdots du_{N} .
\end{equation}
Moreover, Formula~\eqref{eq:Erho} reads as 
\begin{equation}\label{eq:Elin} 
 \E(Z\mid f = a) \, \rho_f(a) = \int_{\R^{N-1}} Z\, \frac{\rho}{|\partial_{u_1}f|}\, du_2\cdots du_{N} .
\end{equation}

\begin{example}\label{ex:cond_hyperplane}
Consider the linear function 
$f(u) = \sum_{i=1}^N w_i u_i$ for a nonzero $w\in\R^N$. Then 
$H=f^{-1}(a)$ is a hyperplane and $\nabla f= w$. We have by definition
$$
 \rho_f(a) = \frac{1}{\|w\|} \int_H \rho\, dH, \quad 
 \rho_H(u) = \Big(\int_H \rho dH\Big)^{-1}\, \rho(u) .
$$
If $w_1 = \partial_{u_1}f \ne 0$, Formula~\eqref{eq:Elin} gives
\begin{equation*}%\label{eq:Epar}
 \E(Z\mid f = a) \, \rho_f(a) = \frac{1}{|w_1|} \int_{\R^{N-1}} Z \rho \, du_2\cdots du_{N} .
\end{equation*} 
In the special case $f(u)=u_N$, we retrieve the known notion 
of the marginal distribution 
$\rho_{u_N}(a) =  \int_{\R^{N-1}} \rho(u_1,\ldots,u_{N-1},a)\, du_1\cdots du_{N-1}$, 
and the conditional density of $Z$ satisfies 
\begin{equation}\label{eq:marg}
 \E(Z \mid u_N = a) \, \rho_{u_N}(a) 
 = \int_{\R^{N-1}} Z(u_1,\ldots,u_{N-1},a) \rho(u_1,\ldots,u_{N-1},a)\, du_1\cdots du_{N-1} .
\end{equation}
\end{example}

%%%

\begin{example}\label{ex:product_rv}
Consider the product function $f(y) = y_1\ddd y_k$, 
and for nonzero $a\in\R$ the smooth hypersurface 
\begin{equation*}%\label{eq:def_Ca}
 C_a := \{y \in \R^k :  y_1\ddd y_k = a \} .
\end{equation*}
If $\rho$ is the joint density of $y\in\R^k$, 
then the pushforward density $\rho_f$ of the product $f(y)$ satisfies, 
by \eqref{eq:pushf_dens} and Lemma~\ref{le:dH_loc_coord}, that 
\begin{equation}\label{eq:rfa} 
\rho_f(a) = \int_{C_a} \frac{\rho}{\|\nabla f\|}\, dC_a = 
\int_{\R^{k-1}} \frac{\rho}{\partial_{y_1}f} \, dy_2\cdots dy_k = 
\int_{\R^{k-1}} \rho\, \frac{dy_2}{|y_2|} \cdots \frac{dy_k}{|y_k|} ,
\end{equation}
since $\partial_{y_1}f= y_2\cdots y_k$. 
We also note that 
$\|\nabla f(y)\| = |a| (\sum_{i=1}^k y_i^{-2})^{\frac12}$.   
Moreover, \eqref{eq:Erho} combined with Lemma~\ref{le:dH_loc_coord}, 
reads as %noting that $\partial_{y_1}f= y_2\cdots y_k$, 
\begin{equation} \label{eq:CErho}
 \E(Z \mid f = a) \, \rho_f(a) 
 = \int_H Z \frac{\rho}{\|\nabla f\|} \, dC_a 
 = \int_{\R^{k-1}} Z \rho \, \frac{dy_2}{|y_2|} \cdots \frac{dy_k}{|y_k|} . 
\end{equation}
\end{example}

\subsection{Products of Gaussians}\label{se:prod_gauss}

In the sequel, we denote by $\varpi_k$ the density of the product $y_1\ddd y_k$ of independent 
standard Gaussian distributed random variables $y_1,\ldots,y_k$; see~\cite{springer-thompson:70}. 
%We shall need some properties of products of independent Gaussian distributed random variables.
%{\tt Check references, e.g., \cite{springer-thompson:70}. Are the results new?}
%{\tt Need to use a different letter: $\delta$ is reserved for Dirac density.}
According to \eqref{eq:rfa} we have for $a\in\R^*$
\begin{equation}\label{eq:dens_prod}
 \varpi_k(a) = \int_{(y_2,\ldots,y_k)\in\R^{k-1}} \varphi(\frac{a}{y_2\ddd y_k}) \varphi(y_2)\ddd\varphi(y_k)\, 
   \frac{dy_2}{|y_2|}\cdots \frac{dy_k}{|y_k|} ,
\end{equation}
where $\varphi(y) = (2\pi)^{-\frac12} e^{-\frac{y^2}{2}}$ denotes the density of 
the standard Gaussian distribution.
%writing $y_1=a(y_2\cdots y_k)^{-1}$. 

More generally, if $y_i\sim\ND(0,\s_i^2)$ are independent centered Gaussians with variance $\s_i^2$, 
then we may write 
$y_i=\s_i\tilde{y}_i$ with $\tilde{y}_i \sim\ND(0,1)$. 
The density $\rho_f$ of the product 
$f(y) = y_1\ddd y_k = \s_1\ddd\s_k \tilde{y}_1\ddd \tilde{y}_k$ then can be expressed 
via \eqref{eq:TR} as 
\begin{equation}\label{eq:pT}
 \rho_f(a) = \frac{1}{\s_1\ddd\s_k}\, \varpi_k\Big(\frac{|a|}{\s_1\ddd\s_k}\Big) .
\end{equation}

%%%
%\subsection{Growth of $\varpi_k$ around the origin}\label{se:growth}

It is easy to see that the density $\varpi_k$ of the product of $k$~standard Gaussians is unbounded for $k\ge 2$:  
we have $\lim_{a\to 0} \varpi_k(a) = \infty$, which causes some technical problems. 
However, the following lemma states that 
the growth of $\varpi_k$ for $a\to 0$ is slow, which  
will be needed for the proof of Theorem~\ref{th:main}:
more specifically, for guaranteeing the assumption~\eqref{eq:Cbound} 
so that Proposition~\ref{th:EBound}
can be applied to the random linear combination
$R(x) = \sum_{i=1}^m u_i q_i(x) x^{d_i}$, 
where the coefficients~$u_i$ are independent random variables 
with the distribution~$\varpi_{k_i}$. 

%The following lemma on the growth of $\varpi_k(a)$  should be known. 
%We include a proof since we couldn't locate the statement in the literature, 

\begin{lemma}\label{le:d2bounds} 
\begin{enumerate}
\item $\varpi_k$ is monotonically decreasing on $(0,\infty)$ and $\varpi_k(-a)=\varpi_k(a)$. 

\item For $0 < \d \le \frac12$ and $a\in\R^*$ we have 
$\varpi_2(a) \le |a|^{\d-1}$.

\item For $a\in\R^*$ we have $\varpi_k(a) \ \le\  e\; |a|^{\frac{1}{2k}-1}$.
\end{enumerate}
\end{lemma}

\begin{proof}
(1) %For the proof the second assertion, 
Taking the derivative in~\eqref{eq:dens_prod} we obtain, 
using symmetry, that 
$$
\varpi'_k(a) = 2^{k-1}\int_{(y_2,\ldots,y_k)\in\R_+^{k-1}} \varphi'\big(\frac{a}{y_2\ddd y_k}\big)\varphi(y_2)\ddd\varphi(y_k)\, 
   \frac{dy_2}{y_2^2}\cdots \frac{dy_k}{y_k^2}  .
$$
Since $\varphi'(y) \le 0$ for $y\ge 0$, we see that $\varpi'_k(a) \le 0$ for $a >0$. 
It is clear that $\varpi_k(-a)=\varpi_k(a)$. 

(2) By \eqref{eq:dens_prod} we have 
$$
 \varpi_2(a) = \int_\R \varphi\big(\frac{a}{y}\big) \varphi(y) \frac{dy}{|y|} 
 = \frac{2}{2\pi} \int_0^\infty \frac{1}{y} e^{-\frac{a^2}{2y^2}}\, e^{-\frac{y^2}{2}}\, dy ,
$$
which we bound as follows:
\begin{equation}\label{eq:om2}
 \varpi_2(a)  \ \le\ \frac{1}{\pi} \int_0^1 \frac{1}{y} e^{-\frac{a^2}{2y^2}}\, dy 
 + \frac{1}{\pi} \int_1^\infty e^{-\frac{y^2}{2}}\, dy 
 \ \le\  \frac{1}{\pi} \int_0^1 \frac{1}{y} e^{-\frac{a^2}{2y^2}}\, dy + \frac{1}{\sqrt{2\pi}} .  
\end{equation}
Let $0 < \d\le\frac12$. 
Since 
$2 x^{\frac{1-\d}{2}}e^{-x} \le 1$ for $x\ge 0$, we obtain 
$e^{-\frac{a^2}{2y^2}}\le  2^{-\frac12}\, |a|^{\d-1}\, y^{1-\d}$ for all $y>0$. 
%$$
%  e^{-\frac{a^2}{2y^2}} \ \le\ \frac{1}{2}\, |a|^{\d-1}\, y^{1-\d} .
%$$
Integrating, we obtain
$$
 \int_0^1 \frac{1}{y} e^{-\frac{a^2}{2y^2}} \, dy \ \le\ 
    2^{-\frac12} |a|^{\d-1}  \int_0^1 y^{-\d}\, dy  \ = \ 
    2^{-\frac12} \frac{|a|^{\d-1} }{1-\d}\,  \ \le\ 2^{\frac12}  |a|^{\d-1} .
$$
Altogether, we get from \eqref{eq:om2} for $|a| \le 1$, 
$$
 \varpi_2(a) \ \le\  \frac{\sqrt{2}}{\pi} \, |a|^{\d-1} + \frac{1}{\sqrt{2\pi}} 
 \ \le\ \Big( \frac{\sqrt{2}}{\pi}  + \frac{1}{\sqrt{2\pi}} \Big) \, |a|^{\d-1}  \ < \  |a|^{\d-1} .
$$
One can check that 
$a^{1-\d}\varpi_2(a) \le a\varpi_2(a) < 1$ for $a\ge 1$. The assertion follows.

(3) The case $k=1$ follows from (1). 
Suppose now $k\ge 2$. 
We have by \eqref{eq:dens_prod} %and the symmetry of $\varphi$ 
\begin{eqnarray*}
\varpi_k(a) &=&\int_{(y_2,\ldots,y_k)\in\R^{k-1}} \varphi\big(\frac{a}{y_2\ddd y_k}\big)\varphi(y_2)\ddd\varphi(y_k)\, 
   \frac{dy_2}{|y_2|}\cdots \frac{dy_k}{|y_k|}  \\ 
 &=& \int_{(y_3,\ldots,y_k)\in\R^{k-2}}  \left[ \int_{y_2\in\R} \varphi\big(\frac{a}{y_2\ddd y_k}\big)\varphi(y_2) \frac{dy_2}{|y_2|}\right]
   \varphi(y_3) \ddd\varphi(y_k)\, \frac{dy_3}{|y_3|}\cdots \frac{dy_k}{|y_k|} \\ 
 &=& \int_{(y_3,\ldots,y_k)\in\R^{k-2}}  
    \varpi_2\big(\frac{a}{y_3\ddd y_k}\big) \varphi(y_3) \ddd\varphi(y_k)\,\frac{dy_3}{|y_3|}\cdots \frac{dy_k}{|y_k|} .  
\end{eqnarray*}
By item~(2) we can bound this by 
$$
\varpi_k(a) \ \le\  |a|^{\d-1} \Big( \int_{y\in\R} |y|^{-\d} \varphi(y)\, dy\Big)^{k-2}
= |a|^{\d-1} \Big(\E |y|^{-\d}\Big)^{k-2} .
$$
Is well known that 
$$
 \E |y|^{-\d} = \frac{1}{\sqrt{\pi}}\, 2^{-\frac{\d}{2}}\, \Gamma(\frac{1-\d}{2}) \ \le\ 
  \frac{1}{\sqrt{\pi}} \, \Gamma(\frac{1-\d}{2}) . 
$$ 
The Taylor expansion
$\frac{1}{\sqrt{\pi}} \, \Gamma(\frac{1-\d}{2}) = 1 + 0.9819... \cdot\d + O(\d^2)$ 
gives the growth for small $\d$: 
it is straightforward to verify 
that $\frac{1}{\sqrt{\pi}} \, \Gamma(\frac{1-\d}{2}) \le\ 1+ 2\d$  
for $0 <\d \le \frac12$.
%We use now for $0 <\d \le \frac12$ the estimation 
%$$
% \E |y|^{-\d} = \frac{1}{\sqrt{\pi}}\, 2^{-\frac{\d}{2}}\, \Gamma(\frac{1-\d}{2}) \ \le\ 
%  \frac{1}{\sqrt{\pi}} \, \Gamma(\frac{1-\d}{2}) \ \le\ 1+ 2\d .
%$$ 
Setting $\d= 1/(2k)$, we obtain 
$\Big(\E |y|^{-\d}\Big)^{k-2} \le (1+2\d)^k = (1+\frac{1}{k})^k < e$
and assertion follows.
\end{proof}

%%%
\section{The Rice Formula}\label{se:Rice}

\subsection{Outline}

The Rice formula is a major tool in the theory of random fields. 
It gives a concise integral expression for the expected number of zeros 
of random functions. For comprehensive treatments we refer 
to~\cite{adler-taylor:07,azais-wschebor:09}.
 
We are going to apply this formula in the following special situation. 
Let $\R[X]_{\le D}$ denote the finite dimensional space of polynomials of degree at most~$D$ 
in the single variable~$X$. 
We study a family of structured polynomials given by a parametrization 
$\R^N \to \R[X]_{\le D}, u \mapsto F_u(X)$, 
where $F_u(X)$ is a polynomial function in the parameter~$u$ and the variable~$X$. 
In our case of interest, it is the parametrization of polynomials by arithmetic circuits of depth four 
in terms of their parameters.

Here is a rough outline of the method. 
We fix a probability density on the space $\R^N$ of parameters. Its pushforward measure on 
$\R[X]_{\le D}$ defines a class of random polynomial functions $F\colon\R\to\R$. 
(It is common to notationally drop the dependence on the parameter $u$.) 
The number $\#\{ x\in [0,1] : F(x) = 0\}$ of real zeros of $F$ then becomes a random variable,  
whose expectation we wish to analyze. For this, let us assume that for almost all $x\in\R$, 
the real random variable~$F(x)$ has a density, denoted by $\rho_{F(x)}$. 
Moreover, we assume that the conditional expectation 
$\E\left( |F'(x)| \mid F(x) = 0 \right)$ is well defined for almost all $x\in\R$. 
The {\em Rice Formula} states that, under some technical 
assumptions, 
\begin{equation*}%\label{eq:Rice}
 \E \left( \#\{ x\in [0,1]: F(x) = 0\} \right) =
 \int_0^1 \E\left( |F'(x)| \mid F(x) = 0 \right) \rho_{F(x)}(0)\, dx .
\end{equation*}
While the idea behind this formula can be easily explained 
(e.g, see \cite[\S3.1]{azais-wschebor:09}), the rigorous justification 
can be quite hard, especially in case of nongaussian distributions 
that we encounter in our work; compare \cite[Thm.~3.4]{azais-wschebor:09}). 
For this reason, we will rely on a weaker version of the Rice formula,
tailored to our situation, that only claims the inequality $\le$ above, %in \eqref{eq:Rice}, 
but has the advantage of requiring less assumptions. 
%(Compare Aza\"{i}s and Wschebor~\cite[Exercise~3.9, p.~69]{azais-wschebor:09}.)
This is the topic of the next subsection. 
Let us emphasize that we do not attempt to state this weaker version of the Rice formula 
in the greatest generality possible. 

%we spell out the assumptions and give the precise statement 
%of this Rice formula we apply. 

%%%
\subsection{A Rice inequality}\label{se:technicalD}

Let $\R^N\times I \to \R,\, (u,x)\mapsto F_u(x)$ be a polynomial function, where 
$I$ is a compact interval. 
We think of $F$ as a parametrization of structured polynomial functions in the variable~$x$ 
in terms of the parameters $u_1,\ldots,u_N$. 
We assume that for all $x\in I$, the polynomial function
$$
 F(x)\colon\R^N\to\R,\, u\mapsto F_u(x)
$$ 
is not constant and thus 
$\{u\in\R^N : \nabla F(x)(u) = 0\}$ has measure zero.

\begin{example} 
\begin{enumerate}
\item Fix integers $0=d_1< d_2 < \ldots < d_t$. Then 
$F_u(x) := u_1+ u_2 x^{d_2} + \ldots + u_t x^{d_t}$ 
parametrizes sparse polynomials with support $\{d_1,\ldots,d_t\}$. 
Note that for all $x\in\R$, $F(x)$~is a nonconstant linear function 
(of the argument~$u$). In particular, $F(x)$ does not have singular values. 

\item Fix integers $0=d_1< d_2 < \ldots < d_t$ and 
$0=e_1< e_2 < \ldots < e_t$. The family 
$F_{u,v}(x) := (u_1+ u_2 x^{d_2} + \ldots + u_t x^{d_t})(v_1+ v_2 x^{e_2} + \ldots + v_t x^{e_t})$
parameterizes products of two sparse polynomials with supports given by 
$\{d_1,\ldots,d_t\}$ and $\{e_1,\ldots,e_t\}$. The set of singular points of $F(x)$ 
consists of the pairs $(u,v)$ such that 
$u_1+ u_2 x^{d_2} + \ldots + u_t x^{d_t}=0, v_1+ v_2 x^{e_2} + \ldots + v_t x^{e_t} =0$.
Thus, for all $x\in\R$, $F(x)$ is surjective and $0$ is its only singular value. 
(We generalize this example in Lemma~\ref{le:SingF}.)
\end{enumerate}
\end{example}

Following Section~\ref{se:cond-exp}, 
if a probability distribution with a density~$\rho$ is given on the space~$\R^N$ of parameters, 
for all $x\in I$,
$F(x)$ becomes a random variable with a well-defined density~$\rho_{F(x)}$.

The following ``Rice inequality'' is the version of Rice's formula that we apply in this paper. 
It is essentially Aza\"{i}s and Wschebor~\cite[Exercise~3.9, p.~69]{azais-wschebor:09}.
We state it in a way that makes the method convenient to apply in our setting. 
We provide the proof for lack of a suitable reference. 
%that it becomes easy to apply in our setting.
%For lack of a suitable reference we provide the proof. 

\begin{thm}\label{th:rice}
Let $\R^N\times [x_0,x_1] \to \R,\, (u,x)\mapsto F_u(x)$ be a smooth function 
such that, for all $x\in [x_0,x_1]$, 
$\{u\in\R^N : \nabla F(x)(u) = 0\}$ has measure zero.
Moreover, we assume that, for almost all $u\in\R^N$, 
the function $[x_0,x_1] \to\R$ has only finitely many zeros. 
Further, let a probability density~$\rho$ be given on~$\R^N$.
We assume there exists an integrable function $g\colon [x_0,x_1] \to [0,\infty]$ 
and $\e>0$ such that for all $x \in [x_0,x_1]$ 
and almost all $a\in (-\e,\e)$ we have 
$$
 \E(|F'(x)| \mid F(x) = a)\, \rho_{F(x)}(a) \ \le\ g(x) .
$$
Then, for a random $u$ with the density $\rho$, we can bound the expected number of zeros 
of the random function $x\mapsto F_u(x)$ in the interval $[x_0,x_1]$ as follows: 
$$
 \E \left( \#\{ x\in [x_0,x_1] : F(x) = 0\} \right)  \ \le\  \int_{x_0}^{x_1} g(x)\, dx. 
$$
\end{thm}

The starting point for the proof of Theorem~\ref{th:rice} 
is Kac's counting formula~\cite[Lemma~1 and Remark~1]{kac:43}. 
A turning point of function is a point where its derivative changes sign.

\begin{lemma}%[Kac counting formula]
\label{pro:kac}
A $C^1$ function $f\colon [x_0,x_1] \to \R$ with only finitely many turning points 
satisfies
$$
 N(f) := \#\{ x\in (x_0,x_1) : f(x) = 0\} \ \le\   
\lim_{\d\to 0} \frac{1}{2\d} \int_{x_0}^{x_1} \ind{\{|f(x)| < \d\}} \;|f'(x)| \, dx .
$$  
In fact, for sufficiently small $\d>0$, the right-hand side equals $N(f)+\eta$, 
where $\eta=0,\frac12,1$ according to as none, one or both of the numbers $x_0,x_1$
are zeros of $f$. 
\end{lemma}

%Instead, we shall use the following weaker version, which 
%can be found in \cite[Exercise~3.1, p.~67]{azais-wschebor:09}.
%The usual version of Kac's counting formula~\cite[Lemma~1 and Remark~1]{kac:43}.

%\cite[Lemma~3.1, p.~55]{azais-wschebor:09}
%gives an equality, but requires that $f$ does not have multiple zeros. 
%We note that this is merely for convenience, since in our setting it is easy to check that 
%for almost all~$u$, $F_u$ does not have multiple zeros. 
% see~\cite[Lemma~3.1, p.~55]{azais-wschebor:09}.
%Instead, we shall use the following weaker version, which 
%can be found in \cite[Exercise~3.1, p.~67]{azais-wschebor:09}.
%The usual version of Kac's counting formula~\cite[Lemma~3.1, p.~55]{azais-wschebor:09}
%gives an equality, but requires that $f$ does not have multiple zeros. 
%We note that this is merely for convenience, since in our setting it is easy to check that 
%for almost all~$u$, $F_u$ does not have multiple zeros. 

%In the following we write $N(f) := \#\{ x\in [0,1] \mid f(x) = 0\}$. 

\begin{proof}[Proof of Theorem~\ref{th:rice}]
In the setting of this theorem, 
we apply Lemma~\ref{pro:kac} to $f=F_u$ for a random $u\in\R^N$. 
Taking expectations over $u$ and using Fatou's lemma, we obtain 
(for convenience, we drop the index~$u$)
$$
 \E(N(F)) \ \le\  \liminf_{\d\to 0} \ \E\Big( \frac{1}{2\d} \int_{x_0}^{x_1} \ind{\{|F(x)| < \d\}} \;|F'(x)| \, dx \Big) .
$$
Due to Tonelli's lemma (nonnegative integrands), 
we can interchange the integral over $x$ and the expectation. 
We obtain
$$
 \E(N(F)) \ \le\  \liminf_{\d\to 0} \int_{x_0}^{x_1} J_\d(x)\, dx ,
$$ 
where we have put
$$
 J_\d(x) := \frac{1}{2\d}\;\E \Big(\ind{\{ |F(x)| < \d\} } \;|F'(x)|\Big) .
$$
Proposition~\ref{pro:coarea} gives for $x\in (x_0,x_1)$, 
$$
  J_\d(x) = \frac{1}{2\d}\; \int_{-\d}^{\d} \E \Big( |F'(x)|\;\mid\; |F(x)| = a \Big) \rho_{F(x)}(x)\, da .
$$
By assumption, the integrand is upper bounded by $g(x)$ for almost all $a \in (-\e,\e)$, hence 
we obtain $J_\d(x) \le g(x)$ for $\d < \e$. Therefore,
$$
 \E(N(F)) \ \le\  \liminf_{\d\to 0} \int_{x_0}^{x_1} J_\d(x)\, dx \ \le\ \int_{x_0}^{x_1} g(x)\, dx .
$$ 
Finally, 
$\E(\#\{ x\in (x_0,x_1) : f(x) = 0\}) = \E(N(F))$ since 
$F(x_0)=0$ and $F(x_1)=0$ happens with probability zero.
\end{proof}

\section{Conditional expectations of random linear combinations}\label{se:BCE}

Throughout, we assume that $u_1,\ldots,u_m$ are independent real random variables 
having the densities $\varphi_1,\ldots,\varphi_m$, respectively. 
We fix real weights $w_1,\ldots,w_m$,
not all being zero, and study the random variable
$$
 f := w_1 u_1 + \cdots + w_m u_m .
$$ 
We shall study bounds for the quantity
$\E \big( |u_i| \mid f =a \big) \rho_f(a)$.  
Since $\nabla f = w \ne 0$, there is no singular value of $f$.
%(Note that the map $u\to f$ is surjective and has no singular values.) 

We begin with a simple bound on the density $\rho_f$ of $f$. It is only useful if the 
densities $\varphi_i$ are bounded (which is not the case for $\varphi=\varpi_k$). 

\begin{lemma}\label{le:rhoZbound}
Suppose that $\|\varphi_i\|_\infty \le A$ for all $i$. 
Then $\|\rho_f\|_\infty \ \le\  \frac{A}{\max_i|w_i|}$. 
In particular, we have $\|\rho_f\|_\infty \le A$ if $w_i=1$ for some~$i$.
\end{lemma}

\begin{proof}
For $a\in\R$ we have by \eqref{eq:rholin} 
\begin{equation*}%\label{eq:rhoeq}
 \rho_f(a) = \frac{1}{|w_1|}\, \int_{\R^{k-1}} \varphi_1(w_1^{-1}(a - w_2u_2 -\cdots w_m u_m)) 
     \varphi_2(u_2) \ddd \varphi_m(u_m) du_2\cdots du_m ,
\end{equation*}
which we can bound as 
$$
 \rho_f(a) \ \le\ \frac{A}{|w_1|} \cdot 
   \int_{\R^{k-1}} \varphi_2(u_2)\cdots \varphi_m(u_m) du_2\cdots du_m = \frac{A}{|w_1|} .
$$
Since the same argument works for $w_i$, this finishes the proof. 
\end{proof}

\begin{defi}\label{def:convenient}
We call a probability density $\varphi$ on $\R$ {\em convenient} 
if $\varphi$ is monotonically decreasing on $(0,\infty)$ and 
symmetric around the origin, i.e., 
$\varphi(-u)=\varphi(u)$ for all~$u\in\R$. 
Moreover, we require
$$
  \E_\varphi := \int_\R |u| \varphi(u)\, du < \infty .
$$ 
\end{defi}

Clearly, a distribution with a convenient density $\varphi$ is centered: 
$\int_\R u \varphi(u)\, du = 0$. 
The densities~$\varpi_k$ of the products of independent Gaussian random variables 
provide examples of convenient densities (see Section~\ref{se:prod_gauss}). 
Note that 
$\E_{\varpi_k} = (\E_\varphi)^k \le 1$ with $\varphi$ denoting 
the density of the standard Gaussian distribution.

%Moreover, we have $|u|\varphi(u) \le \frac12$.
%The reason is that for $u>0$, 
%$$
% u\varphi(u) \le \int_0^u\varphi(t)\, dt \le
%  \int_0^\infty\varphi(t)\, dt = \frac12 .
%$$

\begin{lemma}\label{le:trick}
 Let $\varphi$ and $\psi$ be densities on $\R$ 
and assume that $\varphi$ is convenient. Then:
\begin{enumerate}
\item $|u|\varphi(u) \le \frac12$.
\item $\int_\R |u| \varphi(u)\psi(u) du \ \le\  1$.
\end{enumerate}
\end{lemma}

\begin{proof} 
(1) We have  
$u\varphi(u) \le \int_0^u\varphi(t)\, dt \le
  \int_0^\infty\varphi(t)\, dt = \frac12$, 
for $u>0$.

(2) By Fubini, 
\begin{eqnarray*}
 \int_0^\infty u \varphi(u)\psi(u) du &=&
 \int_0^\infty \Big(\int_0^u dt \Big) \varphi(u)\psi(u) du \\
 &=& \int_{0 \le t \le u} \varphi(u)\psi(u) dt\, du = \int_0^\infty \int_t^\infty \varphi(u)\psi(u) du dt .
\end{eqnarray*}
Now we use that $\varphi$ is monotonically decreasing on $(0,\infty)$ to upper bound this by 
$$
 \int_0^\infty \varphi(t) \int_t^\infty \psi(u) du \, dt 
\le \int_0^\infty \varphi(t) dt = \frac12 .
$$
The assertion follows by the symmetry of $\varphi$. 
\end{proof}

\begin{prop}\label{pro:basicEbound}
Consider $f = w_1 u_1 + \cdots + w_m u_m$, where 
$(w_1,\ldots,w_m)\ne 0$. 
If the density $\varphi_i$ of~$u_i$ is convenient, then we have 
for any $a\in\R$
\begin{equation*}\label{eq:BD1}
  |w_i|\, \E \big( |u_i| \mid f =a  \big) \rho_f(a) \ \le\ 1 .
\end{equation*}
%Suppose the densities of $u_1,\ldots,u_m$ are convenient 
%and $(w_1,\ldots,w_m)\ne 0$. 
%Then $f = w_1 u_1 + \cdots + w_m u_m$ 
%satisfies for any $i$ and $a\in\R$
%\begin{equation*}\label{eq:BD1}
% |w_i|\, \E \big( |u_i| \mid f =a  \big) \rho_f(a) \ \le\ 1 .
%\end{equation*}
\end{prop}

\begin{proof} 
We begin with a general observation. 
Let $v_1$ and $v_2$ be independent random variables with 
the densitites $\psi_1$ and $\psi_2$ and assume $\psi_1$ to be convenient. 
Consider the sum $g(v_1,v_2) := v_1 + v_2$. By \eqref{eq:Elin} we have for $a\in\R$, 
$$
\E( |v_1| \mid g(v_1,v_2) = a)\, \rho_g(a) = \int_\R |v_1| \psi_1(v_1) \psi_2(a-v_1)\, dv_1       
$$
and Lemma~\ref{le:trick}(2) implies 
$\E \big( |v_1| \mid v_1 + v_2=a \big) \rho_g(a) \ \le\ 1$. 
Applying this observation to $v_1:=w_iu_i$ and $v_2:= \sum_{j\ne i} w_j u_j$ 
yields the assertion.
\end{proof}

We provide now another bound on the conditional expectation, which is better 
for small weights $w_i$. For this we need a stronger assumption on the densities.
We will have to deal with unbounded densities, 
namely with the density $\varpi_k$ of the product of $k\ge 2$ standard 
Gaussian random variables. 
Lemma~\ref{le:d2bounds} will allow us to apply the following result to these densities. 

\begin{prop}\label{pro:BD2}
Suppose $u_i$ has a convenient density $\varphi_i$ with $\E_{\varphi_i} \le B$, 
for $i=2,\ldots,m$. 
Further, assume the density $\varphi_1$ of $u_1$ satisfies 
\begin{equation}\label{eq:Cbound}
 \forall u\ \varphi_1(u) \le C\, |u|^{\d -1} 
\end{equation}
for some constants $C>0$ and $0 <\d\le 1$. Then, 
for all $w_2,\ldots,w_{m}\in\R$, the random linear combination 
$f := u_1 + w_2 u_2+ \cdots + w_m u_m$ 
satisfies for $i \ge 2$ and all $a\in\R$,
\begin{equation*}%\label{eq:veryconvenient}
 \E \big( |u_i| \mid f =a\big) \rho_f(a) \ \le\
C\big(\d^{-1} + B\big) |w_i|^{\d-1} .
\end{equation*}
\end{prop}

\begin{proof}
Using the symmetry of $\varphi_i$, 
we can assume w.l.o.g.\ that all the weights $w_i$ are positive.
We first provide the proof in the case $m=2$. 
Let $f = u_1 + wu_2$ with $w>0$. 
By \eqref{eq:Elin} we have 
\begin{equation}\label{eq:ErhoZ0}
 I :=\E ( |u_2| \mid f = a ) \, \rho_f(a) 
 = \int_\R |u_2| \varphi_1(a-wu_2)\varphi_2(u_2) du_2 . 
\end{equation}
By assumption, we have $\varphi_1(a-wu_2) \le C |a-wu_2|^{\d-1}$ for all $u_2\in\R$.
Using this, we obtain
$$
 I \ \le\  C\int_\R |a-wu_2|^{\d-1} |u_2| \varphi_2(u_2)\, du_2 \ \le\ 
 C |w|^{\d-1} \int_\R \Big|\frac{a}{w} - u_2\Big|^{\d-1} |u_2|\varphi_2(u_2)\, du_2 .
$$ 
We bound this integral by splitting according to whether 
$\big|\frac{a}{w} - u_2\big|$ is smaller or larger than one. 
Using that $|u_2|\varphi_2(u_2) \le \frac12$,
which holds since $\varphi_2$ is convenient (see Lemma~\ref{le:trick}(1)), 
we get 
\begin{equation*}\
\begin{split}
 \int_\R \big|\frac{a}{w} - u_2\big|^{\d-1} |u_2|\varphi_2(u_2)\, du_2 \ &\le\ 
 \frac12 \int_{|u_2-a/w| \le 1} \big|\frac{a}{w} - u_2\big|^{\d-1} \, du_2 
 + \int_{|u_2-a/w| \ge 1}  |u_2|\varphi_2(u_2)\, du_2 \\
 &\le\ \frac12 \int_{-1}^1 |x|^{\d-1}\, dx \ +\ \E_{\varphi_2} 
   \le \frac{1}{\d} \ +\  B . 
\end{split}
\end{equation*}
We have thus shown that 
$\E ( |u_2| \mid f = a ) \, \rho_f(a) \le C' |w|^{\d-1} $, 
where  $C':=C(\d^{-1} +  B)$, 
%with $C' := C(\d^{-1} \ +\  \E_{\varphi_2})$,
settling the case $m=2$.  

We now turn to the general case $m\ge 2$.  
Let $f := u_1 + w_2 u_2 + \cdots + w_{m} u_{m}$
and w.l.o.g.\ $i=2$. 
As for \eqref{eq:ErhoZ0}, 
\begin{eqnarray*}
 \lefteqn{\E \big( |u_2| \mid f =a\big) \, \rho_f(a) =} \\
& \int_{\R^{m-1}} \int_\R |u_ 2| \varphi_1( (a-w_3 u_3 - \cdots - w_{m} u_{m}) - w_2u_2)\varphi_2(u_2)\, du_2 \; 
   \varphi_3(u_3)\cdots\varphi_{m}(u_{m})  du_3\cdots du_{m} .
\end{eqnarray*}
We bound the inner integral using the case $m=2$ and obtain,  
$$
 \E \big( |u_2| \mid f =a\big) \, \rho_f(a) \ \le\ 
 C' |w_2|^{\d-1} \int_{\R^{m-1}} \varphi_3(u_3)\cdots\varphi_{m}(u_{m})  du_3\cdots du_{m} 
 =  C' |w_2|^{\d-1} ,
$$
which finishes the proof.
\end{proof}

\section{Random linear combinations of functions}\label{se:zerosLC}

Throughout this section we fix analytic functions $w_1,\ldots,w_m\colon [x_0,x_1]\to\R$
and study for $u\in\R^m$ their linear combination 
$$
 F(x) := \sum_{i=1}^m w_i(x) u_i .
$$
We assume that $w_1,\ldots,w_m$ do not have a common zero in $[x_0,x_1]$. 
Note that $\nabla F(x) = (w_1(x),\ldots,w_m(x)) \ne 0 $ for all~$x$.

%We assume that not all of the $w_i$ are the zero function. 
%xxx

\begin{lemma}\label{le:LIN} 
The set of $u\in\R^m$ such that $\sum_{i=1}^m w_i(x) u_i$ has finitely many zeros 
is of measure zero. %if at least one of the $w_i$ is not the zero function.
\end{lemma}

\begin{proof}
W.l.o.g.\ we can assume that $w_1,\ldots,w_k$ is a basis of the span of $w_1,\ldots,w_m$, 
where $k\ge 1$. 
If we write  
$w_i = \sum_{j=1}^k \lambda_{ij} w_j$, 
for $i> k$ with $\lambda_{ij}\in\R$, then
$F(x) = \sum_{i=1}^m w_i(x) u_i = \sum_{j=1}^k w_j(x) v_j$, 
where $v_j = u_j + \sum_{i=k+1}^m \lambda_{ij} w_j$. 
If the analytic function $F(x)$ has infinitely many zeros in $[x_0,x_1]$, then 
it must vanish identically and thus $v_j = u_j + \sum_{i=k+1}^m \lambda_{ij} w_j =0$ 
for all $j\le k$. Since the set of $u\in\R^m$ satisfying these conditions lie in a lower 
dimensional subspace, the assertion follows. 
\end{proof}

We note that any family of polynomials without common zeros in $[x_0,x_1]$ 
satisfies the above assumptions. 
For instance, we can take the 
family of monomials 
$w_i(x) = x^{d_i}$ with $d_1=0 \le d_2\le \ldots \le d_m$, which 
amounts to studying the random fewnomial 
$F(x) = \sum_{i=1}^m u_i x^{d_i}$.

We assume now that the $u_1,\ldots,u_m$ are independent real random variables 
with the densities $\varphi_1,\ldots,\varphi_m$ and 
consider random linear combination $F(x)$.
(Notationally, we again drop the dependence on $u$.) 
Our goal is to bound the expected number of real zeros of $F$ via the Rice inequality. 

We begin with a simple estimation. 

\begin{prop}\label{cor:Varbound}
Suppose $A,B$ are constants such that 
\begin{equation*}%label{eq:ABbounds}
 \forall i\quad \|\varphi_i\|_\infty \le A, \quad \E_{\varphi_i} \le B .
\end{equation*}
Then $F(x) := u_1 + \sum_{i=2}^m w_i(x) u_i $ satisfies 
for all $x\in [x_0,x_1]$ and $a\in\R$, 
$$
 \E \big(|F'(x)| \mid F(x) = a \big) \rho_{F(x)}(a) \ \le\ 
  A B\, \sum_{i=2}^m |w'_i(x)| .
$$
Moreover, we have  
$$
 \E \#\{x\in [x_0,x_1] : F(x) = 0 \} \ \le\ 
 A B\, \sum_{i=2}^m \int_{x_0}^{x_1} |w'_i(x)| dx.
$$
\end{prop}

\begin{proof}
We have $F'(x) = \sum_{i=2}^m w'_i(x) u_i$, hence 
$|F'(x)| \le \sum_{i=2}^m |w'_i(x)| \cdot |u_i|$.
If we put $w_j := w_j(x)$ for fixed $x$, we have 
\begin{equation*}%\label{eq:Fprimsum}
 \E \big(|F'(x)| \mid F(x) = a \big) \ \le\ 
 \sum_{i=2}^m |w'_i(x)| \cdot \E\big( |u_i| \ \big|\  \sum_{j=1}^m w_j u_j = a\big) .
\end{equation*}
Then $u_2$ and 
$v:=u_1 + w_3u_3+\ldots + w_m u_m$  
are independent random variables and $F(x) = w_2 u_2 + v$.
Let $\vartheta$ denote the density of $v$. 
By Lemma~\ref{le:rhoZbound} we have 
$\|\vartheta\|_\infty \le A$. 
Hence, by~\eqref{eq:Elin}, 
$$
 \E\big( |u_2| \mid F(x) =a)\, \rho_{F(x)}(a) = 
 \int_\R |u_2| \varphi_2(u_2) \vartheta(a -w_2u_2)\, du_2 \ \le\ A \cdot \E_{\varphi_2} \ \le\ AB .
$$
The same bound holds for all $u_i$ with $i\ge 2$ and the first assertion follows. 

By the assumptions on the functions $w_i$ made at the beginning of Section~\ref{se:zerosLC},
we can apply Theorem~\ref{th:rice} and the second assertion follows. 
\end{proof}

The following corollary is of independent interest. 

\begin{cor}
In the situation of Proposition~\ref{cor:Varbound}, 
if all functions $w_i$ are monotonically increasing, then
$$
 \E \#\{x\in [x_0,x_1] : F(x) = 0 \} \ \le\ A B\,\sum_{i=2}^m (w_i(x_1) -w_i(x_0)) .
$$
\end{cor}

In particular, in the case of monomials $w_i(x) = x^{d_i}$ with 
$d_1=0 < d_2 < \ldots < d_m$, the random fewnomial 
$F(x) = \sum_{i=1}^m u_i x^{d_i}$ satisfies
$\E \#\{x\in [0,1] : F(x) = 0 \} \ \le\ A B\, (m-1)$, 
which can be seen as a probabilistic version of Descartes rule. 

\begin{remark} 
Better bounds can be obtained for particular probability distributions of the coefficients $u_i$. 
For instance, one can show that 
$\E \#\{x\in \R : F(x) = 0\} =\Oh(\sqrt{m} \log m)$
for the random sparse polynomial
$F(x) = \sum_{i=1}^m u_i x^{d_i}$ with independent standard Gaussian coefficients; see~\cite{betc:19}. 
\end{remark}
%{\tt State here the $\Oh(\sqrt{m}\log m)$ upper bound of previous version as a remark?} 

%This can be interpreted as a probabilistic version of Descartes rule, 
%giving better bounds than Descartes if $AB<1$. 

Following Proposition~\ref{pro:BD2}, we now provide an estimation, which is better for 
small values of~$w_i(x)$. It is relevant that this does not require the density~$\varphi_i$ 
to be bounded. This estimation can be applied to the distributions 
of products of independent Gaussians, which will be of importance for the proof 
of the main result.

\begin{prop}\label{th:EBound}
Suppose $u_i$ has a convenient density $\varphi_i$ with $\E_{\varphi_i} \le B$, 
for $i=2,\ldots,m$. 
Further, assume there are $C\ge 1$ and $0 <\d\le 1$ such that 
the density $\varphi_1$ of $u_1$ satisfies $\varphi_1(u) \le C\, |u|^{\d -1}$ for all~$u$. 
%We assume that the densities $\varphi_1,\ldots,\varphi_m$ are convenient, 
%$\max_{i \le m}\E_{\varphi_i} \le B$, and that \eqref{eq:Cbound} is satisfied with the constant $C$.
Then, for all $w_2,\ldots,w_m\in\R$, the random linear combination 
$F(x) := u_1 + \sum_{i=2}^m w_i(x) u_i $ satisfies 
%there is a constant $C''$ (depending only on $B,C$) 
for all $x\in [x_0,x_1]$ and all $a\in\R$: 
$$
 \E \big(|F'(x)| \mid F(x) = a \big) \rho_{F(x)}(a) \ \le\ 
  C(\d^{-1} + B) \, \sum_{i=2}^m \frac{|w'_i(x)|}{\max\{ |w_i(x)|, |w_i(x)|^{1-\d}\} }.
$$
\end{prop}

\begin{proof}
Put $C' :=  C(\d^{-1} + B)$. 
Proposition~\ref{pro:BD2} gives for $i \ge 2$, $a\in\R$, and $x\in [x_0,x_1]$, 
$$
 \E \big(|u_i| \mid F(x) = a \big) \rho_{F(x)}(a) \ \le\ \frac{C'}{|w_i(x)|^{1-\d}} .
$$
On the other hand, Proposition~\ref{pro:basicEbound} gives 
$$
 \E \big(|u_i| \mid F(x) = a \big) \rho_{F(x)}(a) \ \le\ \frac{1}{|w_i(x)|} .
$$
Therefore, since $C' \ge C \ge 1$,
% if we put $C'':=\max\{C',1\}$, 
$$
 \E \big( |u_i| \mid F(x) = a \big) \rho_{F(x)}(a) \ \le\ 
 C'\, \frac{1}{\max\{ |w_i(x)|, |w_i(x)|^{1-\d} \} } .
$$
The assertion follows now with $|F'(x)| \le \sum_{i=2}^m |w'_i(x)| |u_i|$. 
\end{proof}

In order to make effective use of Proposition~\ref{th:EBound} 
for certain structured weight functions having product form, 
we introduce the following notion, related to the total variation
$\int_0^1 |q'(x)| dx$ of a function $q$.

\begin{defi}\label{def:def-LV}
The {\em logarithmic variation} of a function 
$q\colon [x_0,x_1] \to (0,\infty)$ is defined as 
\begin{equation*}
\LV(q) := \int_{x_0}^{x_1} \left|\frac{d}{dx} \ln q(x) \right| \, dx = \int_{x_0}^{x_1} \frac{|q'(x)|}{q(x)}\, dx .
\end{equation*}
\end{defi}

The logarithmic variation has the following basic properties, whose proof is obvious.
\begin{lemma}\label{le:new}
\begin{enumerate}
\item If $q$ is monotonically increasing, then 
$\LV(q) = \ln q(x_1) - \ln q(x_0)$. 

\item $LV(q_1\cdot q_2) \le \LV(q_1) + \LV(q_2)$. 

\item $\LV(q^{r}) = |r|\, \LV(q)$ for $r\in\R$. 
%$\LV(q^{-1}) = \LV(q)$.

\end{enumerate}
\end{lemma}

For reasons to become clear in the next section, we assign to a finite subset $S\subseteq\N$ 
the sparse sum of squares with ``support'' $S$ defined as the polynomial 
\begin{equation*}%\label{eq:a_S}
 \a_S(x) := \sum_{s\in S} x^{2s} .
\end{equation*}
We will assume $0\in S$, hence $\a_S(x) \ge 1$ for all $x\in\R$ and $\a_S(0)=1$. 
%If $S\subseteq\N$, then $\a_S(x)$ is a polynomial satisfying 
%$\a_S(x) \ge 1$ for all $x$ and we have $\a_S(0)=1$ iff $0\in S$. 
Moreover, $\a_S(1)=|S|$.

Assume now we have a family of subsets $S_i\subseteq\N$ satisfying $0 \in S_i$ and $|S_i| \le t$, 
for $1\le\ i\le \ell$. We choose 
$1\le k \le \ell$ and define the function 
$$
 q(x) := \Big(\frac{\a_{S_1}(x)\ddd \a_{S_k}(x)}{\a_{S_{k+1}}(x) \ddd \a_{S_\ell}(x)} \Big)^{\frac12} .
$$ 
%For the function $w(x) = q(x) x^d$ with $d\in\N$ defined on the interval $[0,1]$,  
%we can bound the right-hand side of Proposition~\ref{th:EBound} 
%with the following result. 

\begin{prop}\label{pro:new}
Let $d\in\N$ and $0 < \d \le 1$. The function $w\colon [0,1] \to [0,\infty),\, x\mapsto q(x) x^d$ 
satisfies $\LV(q) \le \frac12 \ell \ln t$. 
Moreover, % for $w(x) = q(x) x^d$ with $d\in\N$, we have 
$$
 \int_{0}^{1} \frac{|w'(x)|}{\max\{ w(x), w(x)^{1-\d}\}}\, dx \ \le\ 
  2 \LV(q) + kt + \frac{1}{\d} .
$$
\end{prop}

\begin{proof}
1. By Lemma~\ref{le:new}, we have 
$\LV(\a_{S_i}) \le \ln t$, since $\a_{S_i}$ is monotonically increasing.
Moreover, 
$\a_{S_i}(0)=1$, and $\a_{S_i}(1) \le t$. Again using Lemma~\ref{le:new}, 
we get 
$\LV(q) \le \frac12 \sum_{i=1}^{\ell} \LV(\a_{S_i}) \le \frac12 \ell \ln t$,
showing the first assertion. 

2. We will choose $\e=\e(k,t,d)\in (0,1)$ and bound 
$$
\int_0^1 \frac{|w'(x)|}{\max\{ w(x), w(x)^{1-\d}\}}\, dx  \ \le\ 
 \int_0^\e \frac{|w'(x)|}{w(x)^{1-\d}} \, dx  + \int_\e^1 \frac{|w'(x)|}{w(x)} \, dx  .
$$
For bounding the left-hand integral, we take logarithmic derivatives to get 
from $w(x)=q(x)x^d$ 
\begin{equation}\label{eq:log-der-bd}
 \frac{w'(x)}{w(x)} =  \frac{q'(x)}{q(x)} +  \frac{d}{x} ,
\end{equation}
and hence 
$$
 \frac{w'(x)}{w(x)^{1-\d}} =  \frac{q'(x)}{q(x)} w(x)^{\d}+  \frac{d}{x}  w(x)^{\d} 
 =  \frac{q'(x)}{q(x)} q(x)^{\d} x^{d\d} + d q(x)^{\d} x^{d\d -1} .
$$
For $0\le x \le \e$ we have 
$\a_{S_i}(x) \le 1 + \e^2(t-1) \le 1 + \e^2 t$,  
and hence 
$q(x) \le (1 + \e^2 t)^{\frac{k}{2}}$. 
We can therefore bound 
$$
 \frac{|w'(x)|}{w(x)^{1-\d}} \ \le \ 
  (1 + \e^2 t)^{\frac{k\d}{2}} \, \e^{d\d} \frac{|q'(x)|}{q(x)} 
 + \frac{1}{\d} (1 + \e^2 t)^{\frac{k\d}{2}} \frac{d}{dx}  x^{d\d} .
$$
Integrating over $[0,\e]$, we obtain
\begin{equation*}
\begin{split}
 \int_0^\e \frac{|w'(x)|}{|w(x)|^{1-\d}} \, dx \ &\le\ 
   (1 + \e^2 t)^{\frac{k\d}{2}} \, \e^{d\d} \int_0^\e \frac{|q'(x)|}{q(x)}\, dx 
   + \frac{1}{\d} (1 + \e^2 t)^{\frac{k\d}{2}} \e^{d\d} \\
   &\le\ \big((1 + \e^2 t)^{\frac{k}{2}} \, \e^d\big)^{\d} \big(\LV(q) + \frac{1}{\d} \big) .
\end{split}
\end{equation*}
We now choose $\e:=e^{-\frac{kt}{d}}$. 
Then $\e^d = e^{-kt}$ and 
$$
  (1 + \e^2 t)^{\frac{k}{2}} \e^d \ \le\  (1+t) ^{\frac{k}{2}} \e^d \ \le\ e^{\frac{kt}{2}} \e^d = e^{-\frac{kt}{2}} \ \le\  1 .
$$ 
With this choice of $\e$, we therefore have 
$$
 \int_0^\e \frac{|w'(x)|}{|w(x)|^{1-\d}} \, dx \ \le\ \LV(q) + \frac{1}{\d} .
$$
We next bound the integral over $[\e,1]$, again using \eqref{eq:log-der-bd}, 
$$
 \int_\e^1 \frac{|w'(x)|}{w(x)} \, dx  \ \le\ 
 \int_\e^1 \frac{|q'(x)|}{q(x)} \, dx  +d \int_\e^1 \frac{dx}{x} \ \le\ 
  \LV(q) + d \ln\frac{1}{\e} 
 =  \LV(q) + kt , 
$$
where we used $d \ln\frac{1}{\e} = kt$ by our choice of $\e$. 
Altogether, we obtain 
$$
 \int_0^1 \frac{|w'(x)|}{\max\{ w(x), w(x)^{1-\d}\}} \, dx  \ \le\ 
  \LV(q) + \frac{1}{\d} +  \LV(q) + kt \ \le\ 2\LV(q) + kt + \frac{1}{\d} 
$$
completing the proof.
\end{proof}

\section{Sum of products of sparse polynomials}
\label{se:GSPSP}

Let us first fix some notation. 
We assign to a finite subset $S\subseteq\Z$ of exponents 
and a collection of coefficients~$u_s$, for $s\in S$,  
the Laurent polynomial 
$$
 f_S(x) := \sum_{s\in S} u_sx^s.
$$
Note that 
$f_S(x^{-1}) = f_{-S}(x)$ and 
$f_{d+S}(x) = x^d f_{S}(x)$ for $d\in\Z$. 
This allows to achieve a normalization by shifting exponents: 
let $d$ be the minimum of $S$ 
and put $S':=S-d$. Then 
$S'\subseteq\N$ and $0 \in S'$. Since
$f_{S}(x) = x^d f_{S'}(x)$, 
the functions $f_S$ and $f_{S'}$ have the same number of 
nonzero roots.

Let now $k_1,\ldots, k_m$ and $t$ be positive integers
and fix supports $S_{ij}\subseteq\Z$ 
for $1\le i\le m$ and $1 \le j \le k_i$ such that $|S_{ij}| \le t$.
%we fix supports $S_{ij}\subseteq\N$ with $|S_{ij}| \le t$.
%% for the $t$-sparse polynomials $f_{ij}$. 
We study the number of nonzero real roots of the sum of products
$\sum_{i=1}^m f_{i1}\ddd f_{ik_i}$, 
where $f_{ij} := f_{S_{ij}}$. 
 
By shifting exponents, we assume without loss of generality 
$$
\forall i,j\quad S_{ij} \subseteq\N ,\quad 0\in S_{ij} \quad \mbox{and}\quad |S_{ij}| \le t ,
$$
and consider 
\begin{equation}\label{eq:def-F}
 F(x) := \sum_{i=1}^m f_{i1}(x)\ddd f_{ik_i}(x) x^{d_i} .
\end{equation}
where we allow for a degree pattern $0=d_1 \le d_2 \le \ldots \le d_m$ consisting of natural numbers~$d_i$.

The probabilistic setting is as follows. 
For each $i,j$ and $s\in S_{ij}$ we fix a convenient probability density $\varphi_{ijs}$ on $\R$ 
and assume that there are constants $A,B$ such that 
\begin{equation*}%\label{eq:ABbounds}
 \forall i,j,s\quad \|\varphi_{ijs}\|_\infty \le A, \quad \E_{\varphi_{ijs}} \le B .
\end{equation*}
We suppose that we have random univariate polynomials 
\begin{equation}\label{eq:def_fij}
 f_{ij}(x) = \sum_{s\in S_{ij}} u_{ijs} x^s
\end{equation}
with independent real coefficients $u_{ijs}$ having the convenient density $\varphi_{ijs}$. 
The goal is to study the expected number of real zeros of the resulting random polynomial~$F$. 

We assign to the support $S_{ij}$ the following generating functions
\begin{equation}\label{eq:gen-fct}
\a_{ij}(x) := \sum_{s\in S_{ij}} x^{2s} , \quad 
 \b_{ij}(x) := \sum_{s\in S_{ij}} x^s .
\end{equation}
Note that $\E (f_{ij}(x)^2) = \a_{ij}(x)$ if $\E (u_{ijs}^2) =1$, since 
$\E(u_{ijs})=0$.

%{\tt Discuss set of singular points of the surjective map $(u,x) \to F(u,x)$. 
%It is union of subspace arrangements of codim two. Discuss first the case $m=1$.}

%We are going to verify that the assumptions of Subsection~\ref{se:technicalD} 
%are satisfied in our situation. Here is the first step. 

The next lemma makes sure we can apply Theorem~\ref{th:rice} in the above setting. 

\begin{lemma}\label{le:SingF}
Let $N:= \sum_{i=1}^m \sum_{j=1}^{k_i} |S_{ij}|$ denote the number of parameters.  
For $x\in\R$ consider the polynomial map $F(x)\colon\R^N\to\R$
sending a system $u=(u_{ijs})\in\R^N$ of coefficients to $F(x)$, 
as defined in~\eqref{eq:def-F}. For all $x\in\R$ we have: 
\begin{enumerate}
\item[(a)] $F(x)$ is surjective and thus nonconstant. 
\item[(b)] All nonzero $a\in\R$ are regular values of $F(x)$. 
\item[(c)] $0$ is a singular value of $F(x)$ unless $k_1=\ldots = k_m=1$.
\end{enumerate}
The conditional density $\rho_{F(x)}(a)$ is defined at every nonzero $a\in\R$. 
However, it is undefined at $a=0$, unless $k_1=\ldots k_m=1$.
\end{lemma}

\begin{proof} We fix $x\in\R$. 
(a) After specializing $u_{ijs}:=0$ for $s\ne 0$, $F(x)$ becomes the function 
mapping $(u_{ij0})$ to $\sum_{i=1}^m u_{i10} \ddd u_{ik_i0}$, which clearly is a surjective function. 

(b) The $f_{ij}(x)$ are linear functions in disjoints sets of variables and all have a nonzero coefficient.
Therefore, their gradients, viewed as vectors in $\R^N$, are linearly independent. 
Suppose now $u=(u_{ijs})\in\R^N$ is a singular point of $F(x)$. We have 
(dropping the argument~$u$) 
$$
 \nabla F(x) = \sum_{i=1}^m f_{i,1}\ddd f_{i,j-1}\nabla f_{i,j} f_{i,j+1}\ddd f_{i,k_i} .
$$ 
Since the $\nabla f_{i,j}$ are linearly independent, we must have 
$f_{i,1}\ddd f_{i,j-1}\nabla f_{i,j} f_{i,j+1}\ddd f_{i,k_i} =0$ for all $i,j$. 
This means that for all $i$ there are different $j$ and $j'$ such that 
$f_{ij}(x)=0$ and $f_{ij'}(x)=0$.  
In particular, we have $F(x)(u)=0$ for such $u$ and hence 
$0$ is the only possible singular value of $F(x)$. 
If $k_i>1$ for some~$i$, then $u=0$ is a singular point of $F(x)$ 
and thus $0$ is a singular value. 

(c) This follows from the reasoning in (b).
\end{proof}

For applying Theorem~\ref{th:rice}, 
the main work consists now in exhibiting a ``small'' integrable function $g(x)$ 
that upper bounds  the conditional expectations. 
We embark on this next.

%%%
\subsection{Products of sparse polynomials}

We analyze here the case $m=1$ of one product
$$
 g(x) := f_1(x)\ddd f_k(x) 
$$
of random $t$-sparse polynomials 
$f_j(x) = \sum_{s\in S_j} u_{js} x^s$, where for convenience, we drop the 
index $i=1$. In particular, we write 
$\b_{j}(x) := \sum_{s\in S_{j}} x^s$. 
So we assume $0\in S_j$ and $|S_j| \le t$ for all $j$. 

By Lemma~\ref{le:SingF}, every nonzero $a\in\R$ is a regular value of  
the map $g(x)\colon\R^N\to\R$, thus the conditional density $\rho_{g(x)}(a)$ 
is well defined and so are the conditional expectations with respect to the 
condition $g(x) = a$, provided $\rho_{g(x)}(a)>0$.  

\begin{lemma}\label{le:PProducts}
For all $x\in\R$ and all nonzero $a\in\R$ we have 
$$
 \E\Big( \left|\frac{f'_j(x)}{f_j(x)} \right| \mid g(x) = a \Big) \rho_{g(x)}(a) \ \le\  AB \,  \frac{\b'_j(x) }{|a|} ,
$$
$$
 \E (|g'(x)| \mid g(x) = a ) \rho_{g(x)}(a) \ \le\ AB\,\sum_{j=1}^k \b'_j(x) .
$$
\end{lemma}

\begin{proof}
Fix $x\in\R$ and   
consider the random variables $y_j:=f_j(x)$ and $z_j:=f'_j(x)$.
%{\tt Conflict of notation with $u_{ijk}$!} 
If $\psi_j(y_j,z_j)$ denotes the joint density of $(y_j,z_j)$, then 
by the independence of $(y_1,z_1),\ldots,(y_k,z_k)$, 
the probability density of $(y,z)\in\R^k\times\R^k$ is given by 
 $\psi_1(y_1,z_1)\cdot\ldots\cdot \psi_k(y_k,z_k)$.
%$\psi_j(y_j,z_j)$ denotes the joint density of $(y_j,z_j)$.
Note that 
$g(x)=y_1\ddd y_k$.

We are going to apply some insights from Section~\ref{se:prelim}. %{se:prod_gauss}.
Namely, we apply Equation~\eqref{eq:Erho} to the function 
$f\colon\R^k\times\R^k\to\R,\, (y,z) \mapsto y_1\cdot\ldots\cdot y_k$ 
and the random variable $Z(y,z) := |\frac{z_1}{y_1}|$.
For nonzero $a\in\R$ we consider the hypersurface 
$C_a := \{ y \in\R^k : y_1\cdots y_k = a\}$ 
and note that 
$\|\nabla(y_1\ddd y_k)\| = |a| (\sum_{i=1}^k y_i^{-2})^{\frac12}$.
We obtain 
\begin{equation}\label{eq:OV}
\begin{split}
  & \E\left(\left| \frac{z_1}{y_1} \right| \ \Big|\ y_1\ddd y_k = a \right) \rho_{g(x)}(a) \\
  &= \int_{C_a\times\R^k} \left| \frac{z_1}{y_1} \right| \psi_1(y_1,z_1) \ddd\psi_k(y_k,z_k)\,
   \frac{d(C_a\times\R^k)}{|a| \big(\sum_{i=1}^k y_i^{-2}\big)^{\frac12}} \\
  &= \int_{y\in C_a} \left[ \int_{z\in\R^k} \left| \frac{z_1}{y_1} \right| 
 \psi_1(y_1,z_1) \ddd\psi_k(y_k,z_k)\, dz_1\cdots dz_k\, \right] \frac{dC_a}{|a| \big(\sum_{i=1}^k y_i^{-2}\big)^{\frac12}} . 
\end{split}
\end{equation}
For fixed $y\in C_a$, 
the inner integral can be simplified to
\begin{equation*}
\begin{split}
  &  \int_{z_1\in\R} \left| \frac{z_1}{y_1} \right| \psi_1(z_1,y_1)\, \left[
  \int_{(z_2,\ldots,z_k)\in\R^{k-1}} \psi_2(y_2,z_2)\ddd \psi_k(y_k,z_k)\, dz_2\cdots dz_k\ \right] \, dz_1 \\
  =& \int_{z_1\in\R} \left| \frac{z_1}{y_1} \right| \psi_1(y_1,z_1)\, dz_1\,  \cdot \psi_2(y_2)\ddd\psi_k(y_k) ,
\end{split}
\end{equation*}
with the marginal densities $\psi_i$ defined by 
$\psi_i(y_i) := \int_\R \psi_i(y_i,z_i) dz_i$. 
By~\eqref{eq:marg}
we have for $y_1\in\R^*$, 
% {\tt more details? Cp.~eq.(2), p.1 in my notes of 18.2.13.} 
$$
 \int_{z_1\in\R} \left| \frac{z_1}{y_1} \right| \psi_1(y_1,z_1) \, dz_1 
 = \E \Big(\left| \frac{z_1}{y_1} \right| \ \Big|\ y_1 \Big) \psi_1(y_1) .
$$
We thus obtain from \eqref{eq:OV}
\begin{equation}\label{eq:FOUR}
\begin{split}
  & \E\left(\left| \frac{z_1}{y_1} \right| \ \Big|\ y_1\ddd y_k = a \right) \rho_{g(x)}(a) \\
  & = \int_{y \in C_a} \E\Big(\left| \frac{z_1}{y_1} \right| \ \Big|\ y_1 \Big) \psi_1(y_1) \, \psi_2(y_2)\ddd \psi_k(y_k)\, 
                 \frac{dC_a}{|a| \big(\sum_{i=1}^k y_i^{-2}\big)^{\frac12}} . 
\end{split}
\end{equation}

Proposition~\ref{cor:Varbound} applied to the random linear combination 
$f_j(x) = \sum_{s\in S_j} u_{js} x^s$ implies %(here we essentially use $0\in S_j$)
$$
 \E (|z_1| \mid y_1)\, \psi_1(y_1) \ \le\ AB\,  \b'_1(x) . 
$$ 
Here we essentially use that, 
due to the assumption $0\in S_j$,
the polynomial $f_j(x) = u_{j0} + \ldots$ has a constant term.  
Using this bound, we get from \eqref{eq:FOUR},
\begin{equation*}
\begin{split}
 & \E\left(\left|\frac{z_1}{y_1} \right| \ \Big|\ y_1\ddd y_k = a \right) \rho_{g(x)}(a) \\
& \le\  AB\,  \cdot \b'_1(x) 
     \int_{y \in C_a} \frac{1}{|y_1|} \psi_2(y_2) \ddd \psi_k(y_k)\, 
          \frac{dC_a}{|a| \big(\sum_{i=1}^k y_i^{-2}\big)^{\frac12}} .%d\tilde{C}_a .
\end{split}
\end{equation*}
Using~\eqref{eq:CErho}, the  integral over $C_a$ simplifies to
\begin{eqnarray*}
\lefteqn{\int_{(y_2,\ldots,y_k)\in\R^{k-1}} \frac{|y_2\ddd y_k|}{|a|} \cdot \psi_2(y_2) \ddd \psi_k(y_k) 
  \frac{dy_2\cdots dy_k}{|y_2|\cdots |y_k|} } \\
 &=& \frac{1}{|a|} \int_\R \psi_2(y_2) dy_2 \ddd \int_\R \psi_k(y_k) dy_k  = \frac{1}{|a|} .
\end{eqnarray*}
%where we used that (see \eqref{eq:dCtilde})
%$$
% d\tilde{C}_a = \frac{du_2}{|u_2|} \cdots \frac{du_k}{|u_k|} .
%$$
%%(cf. eq.~(4), p.3 of notes of 18.2.13). 
Therefore, indeed
$$
 \E \Big(\left|\frac{f'_1(x)}{f_1(x)} \right| \ \Big|\ g(x) = a \Big)  \rho_{g(x)}(a)\ \le\  
  AB \, \frac{\b'_1(x) }{|a|} .
$$
The same argument works with $f_j$ instead of $f_1$, 
so that we have proved the first statement. 

In order to show the second statement, 
taking logarithmic derivatives, we get 
$$
 \frac{g'(x)}{g(x)} = \sum_{j=1}^k \frac{f'_j(x)}{f_j(x)} ,
$$
hence 
$$
 \left|\frac{g'(x)}{g(x)} \right| \ \le\  \sum_{j=1}^k \left|\frac{f'_j(x)}{f_j(x)} \right| .
$$ 
Therefore, 
$$
 \E\Big(\left|\frac{g'(x)}{g(x)} \right| \ \Big|\ g(x) = a \Big)  \ \le\  
  \sum_{j=1}^k \E\Big(\left|\frac{f'_j(x)}{f_j(x)} \right| \Big| g(x) = a \Big) 
$$
hence 
$$
 \E\Big( |g'(x)| \mid g(x) = a \Big) \rho_{g(x)}(a) \ \le\  
  \sum_{j=1}^k |a|\, \E\Big(\left|\frac{f'_j(x)}{f_j(x)} \right| \ \Big|\ g(x) = a \Big) \rho_{g(x)}(a) .
$$
Inserting here the bound of the first statement yields the second statement.
\end{proof}

%%%
\subsection{Polynomials with nonzero constant coefficient}
\label{se:SPSP}

We deal here with the special case $d_1=\ldots =d_m=0$.  
So we are in the situation where all the $f_{ij}$ almost surely have a nonzero constant coefficient.
It turns out that this situation is way easier to analyze than the general case. 

The next result shows that the real $\tau$-conjecture is true on average 
under the assumption $d_1=\ldots =d_m=0$, if we only count zeros in $[0,1]$.
%The technical assumption means that we consider sums of products 
%of random $t$-sparse polynomials having a constant term. 
It is worthwile noting that this results holds for any convenient distribution of the 
coefficients $u_{ijs}$, as long as they are independent. 

%Removing the technical assumption turned out to be a nontrivial task, 
%which we only carried out in the case where all the densities $\varphi_{ijs}$ 
%are standard Gaussian. See the next section for this. 
 
\begin{thm}\label{th:SP1}
Under the assumptions from the beginning of Section~\ref{se:GSPSP}, %(and $d_1=\ldots=d_m=0$), 
the random polynomial 
$F= \sum_{i=1}^m f_{i1}\ddd f_{ik_i}$
satisfies 
$$
 \E \#\{ x\in [0,1] : F(x) = 0\} \ \le\ A B \, (k_1 + \ldots + k_m) (t-1). 
$$
\end{thm}

\begin{proof}
Lemma~\ref{le:SingF} guarantees that $(u,x)\mapsto g(x)(u)$ 
satisfies the assumptions of Theorem~\ref{th:rice}. 

We are going to show that forall $x\in\R$ and all nonzero $a\in\R$,  
\begin{equation}\label{eq:fifi}
 \E( |F'(x)| \mid F(x) = a ) \rho_{F(x)}(a) \ \le\ A B \, \sum_{i=1}^m \sum_{j=1}^{k_i} \b'_{ij}(x),
\end{equation}
where we recall that $\b_{ij}(x)$ was defined in~\eqref{eq:gen-fct}. 
Then, taking into account Lemma~\ref{le:SingF} and 
$\int_0^1 \b'_{ij}(x) dx = \b_{ij}(1) - \b_{ij}(0) \le t -1$, 
the assertion will follow by Theorem~\ref{th:rice}.

Towards proving \eqref{eq:fifi}, we 
put $g_i(x) := f_{i1}(x)\ddd f_{ik_i}(x)$. 
Then we have $F(x) = g_1(x) + \ldots + g_m(x)$ and hence 
$|F'(x)| \le \sum_{i=1}^m |g'_i(x)|$. Therefore, 
$$
 \E (|F'(x)| \mid F(x) = a ) \rho_{F(x)}(a) \ \le\ 
\sum_{i=1}^m  \E (|g'_i(x)| \mid F(x) = a ) \rho_{F(x)}(a) .
$$
%Fix $1\le i \le m$ and $x\in\R$ nonzero. 
Lemma~\ref{le:PProducts} gives 
for nonzero $b\in\R$ that 
\begin{equation}\label{eq:EUG}
 \E (|g'_i(x)| \mid g_i(x) = b ) \rho_{g_i(x)}(b) \ \le\ AB\,\sum_{j=1}^{k_i} \b'_{ij}(x) .
\end{equation}
%(Recall that this lemma requires $0\in S_{ij}$.)
For proving \eqref{eq:fifi}, it suffices to show that the same bound holds 
when conditioning on $F(x)=a$, namely
\begin{equation}\label{eq:claim-g'}
 \E (|g'_i(x)| \mid F(x) = a ) \rho_{F(x)}(a) \ \le\ AB\,\sum_{j=1}^{k_i} \b'_{ij}(x) .
\end{equation}

For showing this, we fix $1\le i \le m$.
We put $y_i := g_i(x)$ and $z_i := g'_i(x)$, and denote by 
$\psi_i(y_i,z_i)$ the joint density of $(y_i,z_i)$. Moreover, 
we write $\psi_i(y_i) := \int_\R \psi_i(y_i,z_i)\, dz_i$ for 
the first marginal distribution. 
By construction, the pairs $(y_1,z_1),\ldots,(y_m,z_m)$ are independent. 
%Let us denote by $\rho_{y_1+\ldots u_m}$ the pushforward densities of $y_1+\ldots u_m$ etc. 
We claim that 
\begin{equation}\label{eq:claim-g-gen}
 \E (|z_1| \mid y_1 + \ldots +y_m = a )\, \rho_{y_1+\ldots +y_m}(a) = 
  \int_\R \E (|z_1| \mid y_1 = b )\, \psi_1(b)\,\rho_{y_2+\ldots + y_m}(a-b)\, db .
\end{equation}
It remains to prove this claim, since it implies, combined with \eqref{eq:EUG}, that 
\begin{equation*}
\begin{split}
& \E \Big(|g'_1(x)| \ \Big|\ \sum_{j=1}^m g_j(x) = a \Big) \rho_{F(x)}(a)   =  
 \int_{b\in\R} \E \big(|g'_1(x)| \ \big|\  g_1(x) = b \big) \rho_{g_1(x)}(b)\, \rho_{\sum_{j\ne 1}g_j(x)}(a-b)\, db \\ 
& \le\ AB\, \sum_{j=1}^{k_i} \b'_{1j}(x) \int_{b\in\R} \rho_{\sum_{j\ne i} g_j(x)}(a-b)\, db = AB\, \sum_{j=1}^{k_i} \b'_{1j}(x) ,
\end{split}
\end{equation*}
which is \eqref{eq:claim-g'} (for w.l.o.g.\ $i=1$). 

We deduce now the claim \eqref{eq:claim-g-gen}. By \eqref{eq:Elin} we have 
\begin{equation}\label{eq:aintegral}
\begin{split}
 &\hspace{6ex}\E (|z_1| \mid y_1 + \ldots +y_m = a )\, \rho_{y_1+\ldots +y_m}(a) \\
 &=\int_{\R^{m-1}} \int_{\R^{m}} |z_1| \psi_1(a-y_2-\ldots -y_m,z_1) \psi_2(y_2,z_2)\cdots\psi_m(y_m,z_m)\, dz_1\cdots dz_m \, dy_2\cdots dy_m \\
 &=\int_{\R^{m-1}} \left[\int_{\R} |z_1| \psi_1(a-y_2-\ldots -y_m,z_1)\, dz_1\, \right] \psi_2(y_2)\cdots\psi_m(y_m)\, dy_2\cdots dy_m .
% &=\int_{\R^{m-1}} \E \big( |z_1| y_1=-y_2-\ldots -y_m,z_1) \psi_1(-y_2-\ldots -y_m)
%\psi_2(y_2)\cdots\psi_m(y_m)\, dy_2\cdots dy_m dz_1 \\
\end{split}
\end{equation}
For fixed $y_2,\ldots,y_m$ and $b = a- y_2 \ldots -\ldots -y_m$, the 
expression in parenthesis equals 
$$
 \E (|z_1| \mid y_1 = b)\, \psi_1(b) .
$$
By applying Proposition~\ref{pro:coarea} to the map 
$T\colon\R^{m-1}\to\R, (y_2,\ldots,y_m) \mapsto a-y_2 -\ldots -y_m$, 
taking into account Lemma~\ref{le:dH_loc_coord}, 
we can express the above integral~\eqref{eq:aintegral} as 
$$
 \int_\R \E (|z_1| \mid y_1 = b)\, \psi_1(b) \left[\int_{T^{-1}(b)} \psi_2(y_2)\cdots\psi_m(y_m)\, \frac{dT^{-1}(b)}{\sqrt{m-1}} \right]\, db .
$$
By definition, the expression in the parenthesis equals the pushforward density 
$\rho_{y_2+\ldots y_m}(a-b)$, 
which shows the claim~\eqref{eq:claim-g-gen} and finishes the proof.
\end{proof}

\subsection{Proof of main result}
\label{se:SPSPgeneral}

We specialize the setting described at the beginning of Section~\ref{se:GSPSP} 
to the case where all the coefficients $u_{ijs}$ are standard Gaussian. 

For $1\le i \le m$ we define the auxiliary analytic weight functions
\begin{equation}\label{def:qi}
 q_i(x) := \prod_{j=1}^{k_i} \left(\frac{\a_{ij}(x)}{\a_{1j}(x)}\right)^{\frac12} ,
\end{equation}
and recall that $\a_{ij}(x)$ was defined in~\eqref{eq:gen-fct}. 
Note that $q_i(x) >0$ for all $x\in\R$ and $q_1(x)=1$. 
We define the analytic weight function 
$w_i(x) := q_i(x) x^{d_i}$ for $1\le i \le m$ and note that $w_1(x)=1$. 

We will reduce the problem of counting the expected number of zeros of the structured 
random polynomial $F(x)$ to the study of the expected number of zeros of random 
linear combinations 
\begin{equation*}%\label{eq:def-R}
 R(x) := \sum_{i=1}^m u_i q_i(x) x^{d_i} = u_1 + u_2 q_2(x)x^{d_2} + \ldots + u_m q_m(x) x^{d_m},
\end{equation*}
of the weight functions~$w_i(x)$, where 
the coefficients $u_i$ are 
independent and follow the distribution $\varpi_{k_i}$ of 
a product of $k_i$ standard Gaussians
(cf.\ Section~\ref{se:prod_gauss}). 
We note that, due to Lemma~\ref{le:LIN}, for almost all  $u\in\R^m$, 
the function~$R$ has only finitely many zeros in $[x_0,x_1]$. 
Thus it satisfies the assumptions stated at the beginning of Section~\ref{se:zerosLC}.

\begin{prop}\label{th:mainbound}
For $x\in\R$ and nonzero $a\in\R$, we have 
\begin{eqnarray*}
 \E\big( |F'(x)| \mid F(x) = a \big) \rho_{F(x)}(a) &\le&\frac{1}{\sqrt{2\pi}} \sum_{i=1}^m \sum_{j=1}^{k_i} \b'_{ij}(x) 
 +\sum_{i=1}^m \left| \frac{q'_i(x)}{q_i(x)} \right| \\
& & +\phantom{x} \E\Big( |R'(x)| \mid R(x) = a \Big) \rho_{R(x)}(a) .
\end{eqnarray*}
\end{prop}

\begin{proof}
We write $F(x) = h_1(x)+\ldots + h_m(x)$, where 
$$
 h_i(x) := g_i(x) x^{d_i} \quad  \mbox{and}\quad  g_i(x) := f_{i1}(x)\ddd f_{ik_i}(x) .
$$
Note that 
$h'_i (x) = g'_i(x) x^{d_i} + g_i(x)  d_i x^{d_i-1}$. 
We bound with the triangle inequality: 
$$
|F'(x)| \ \le\  \sum_{i=1}^m \left| g'_i(x) x^{d_i} \right| 
  +  \left| \sum_{i=1}^mg_i(x) d_i x^{d_i-1} \right| .
$$
Here, it is essential not to upper bound further the right-hand 
contribution by $\sum_{i=1}^m |g_i(x) d_i x^{d_i-1}|$.  
Continuing, we get 
\begin{equation}\label{eq:RHC}
\begin{split}
  \E\big( |F'(x)| \mid F(x) = a \big) \rho_{F(x)}(a) \ \le\ &  
   \sum_{i=1}^m  \E\big( |g'_i(x) x^{d_i}| \mid F(x) = a \big) \rho_{F(x)}(a) \\
  &+\; \E\Big(\big| \sum_{i=2}^m g_i(x) d_i x^{d_i-1} \big| \mid F(x) = a \big) \rho_{F(x)}(a) .
\end{split}
\end{equation}
By the same reasoning as for \eqref{eq:claim-g-gen}, we have for nonzero $a\in\R$
%For fixed $i$, we write 
$$
 \E\big( |g'_i(x) x^{d_i}| \mid F(x) = a \big) \rho_{F(x)}(a) = 
 \int_{\R} \E\big( |g'_i(x) x^{d_i}| \mid h_i(x) = b \big) \rho_{h_i(x)}(b) \rho_{H_i(x)}(a-b)\, db ,
$$
where $H_i(x) := \sum_{j\ne i} h_j(x)$. Moreover, setting 
$\tilde{b} :=  \frac{b}{x^{d_i}}$, we get for nonzero $b\in\R$
$$
 \E\big( |g'_i(x) x^{d_i}| \mid h_i(x) = b \big) \rho_{h_i(x)}(b)
 = \E\big( |g'_i(x) x^{d_i}| \mid g_i(x) = \tilde{b} \big) 
\frac{1}{x^{d_i}}\rho_{g_i(x)}(\tilde{b}) .
$$
The $x^{d_i}$ cancels and by Lemma~\ref{le:PProducts}, we have with 
$AB= \frac{1}{\sqrt{2\pi}}$,  
$$
\E\big( |g'_i(x) x^{d_i} | \mid h_i(x) = b \big) \rho_{h_i(x)}(b) 
\le \frac{1}{\sqrt{2\pi}} \sum_{j=1}^{k_i} \b'_{ij}(x) .
$$
Note that the right hand-side does not depend on~$b$. 
Multiplying with $\rho_{H_i(x)}(a-b)$, integrating over $b$ (which doesn't change anything)
and summing over $i$, yields the first contribution in the Theorem's upper bound. 

It remains to bound the right-hand contribution in \eqref{eq:RHC}, 
for fixed $x\in\R$ and nonzero $a\in\R$. 
For this, note that $f_{ij}(x)$ is a centered Gaussian random variable having 
the variance $\a_{ij}(x)$ (recall \eqref{eq:def_fij} and \eqref{eq:gen-fct}). 
So we may write 
$f_{ij}(x) = \a_{ij}(x)^{\frac12} v_{ij}$ with independent standard Gaussian random variables~$v_{ij}$. 
Hence, if we abbreviate $u_i :=  v_{i1}\cdots v_{ik_i}$ and put 
$p_i(x) := (\a_{i1}(x)\ddd \a_{ik_i}(x))^{-\frac12}$, then 
$$
 g_i(x) = f_{i1}(x)\ddd f_{ik_i}(x) = \a_{i1}(x)^{\frac12} \ddd \a_{ik_i}(x)^{\frac12} v_{i1}\ddd v_{ik_i}  
 = p_i(x)^{-1} u_i .  
$$
By its definition, the random variable $u_i$ has the distribution $\varpi_{k_i}$ 
(cf.\ Section~\ref{se:prod_gauss}). 
It is a convenient distribution (cf.\ Definition~\ref{def:convenient}). 
We also note that 
$q_i(x)=\frac{p_1(x)}{p_i(x)}$
by \eqref{def:qi}. 
%Put $z_i:= (\a_{i1}\cdots \a_{ik_i})^{-\frac12}$ and note that 
%$q_i=\frac{z_1}{z_i}$. 
%We can write $g_i(x) = z_i(x)^{-1} u_i$, where $u_i$ is a random variable with 
%the distribution $\d_{k_i}$. 
%{\tt Too quick! Refer to a result in Section~\ref{se:prod_gauss}?} 
With these notations, we can write 
$$
 F(x) = \sum_{i=1}^m g_i(x) x^{d_i} 
 = \sum_{i=1}^m \frac{u_i}{p_i(x)} x^{d_i} 
 = \frac{1}{p_1(x)} \sum_{i=1}^m u_i q_i(x) x^{d_i}
= \frac{1}{p_1(x)} R(x) .
$$
Hence 
$\rho_{F(x)}(a) = p_1(x) \rho_{R(x)}(\z)$, where $\z := p_1(x)a$.  
We analyze now the right-hand contribution in \eqref{eq:RHC}: 
$$
 \E\Big(\big| \sum_{i=2}^m \frac{u_i}{p_i(x)} d_i x^{d_i-1} \big|\mid F(x) = a \Big) \rho_{F(x)}(a) 
= \E\Big(\big| \sum_{i=2}^m u_i q_i(x) d_i x^{d_i-1} \big|\mid F(x) = a \Big) \rho_{R(x)}(\z) .
$$
Using
$$
 R'(x) = \sum_{i=1}^m u_i q'_i(x) x^{d_i} + \sum_{i=1}^m u_i q_i(x) d_i x^{d_i-1} , 
$$
we can bound 
$$
 \Big|\sum_{i=2}^m u_i q_i(x) d_i x^{d_i-1} \Big| \ \le\ 
\sum_{i=1}^m \Big| u_i q'_i(x) x^{d_i} \Big| + | R'(x)| .
$$ 
Therefore,
\begin{eqnarray*}
 \E\Big( \big|\sum_{i=2}^m u_i q_i(x) d_i x^{d_i-1} \big| \mid R(x) = \z \Big) \rho_{R(x)}(\z) 
  &\le& \sum_{i=1}^m \E(\big| u_i q'_i(x) x^{d_i} \big| \mid R(x) = \z \big) \rho_{R(x)}(\z)  \\
 &+& \E\big( |R'(x)| \mid R(x)=\z \big) \rho_{R(x)}(\z) .\\
\end{eqnarray*}
Note that the right-hand contribution equals 
$\E\big( |R'(x)| \mid R(x)=a\big) \rho_{R(x)}(a)$
as desired. In order to bound the left-hand sum, 
we can apply Proposition~\ref{pro:basicEbound} since the densities of the $u_i$ are convenient,  
and we thus obtain 
$$
 |q_i(x) x^{d_i}| \cdot \E( |u_i| \mid R(x) = \z \big) \rho_{R(x)}(\z)  \ \le \ 1 .
$$
This yields 
$$
 |q'_i(x) x^{d_i}| \cdot \E(\big|u_i \big| \mid R(x) = \z \big) \rho_{R(x)}(\z) \ \le \ 
 \frac{|q'_i(x)|}{q_i(x)} .
$$
Summarizing, we have shown that 
$$
 \E\Big(\big| \sum_{i=2}^m g_i(x) d_i x^{d_i-1} \big| \mid F(x) = a \big) \rho_{F(x)}(a) \ \le\ 
   \sum_{i=1}^m \frac{|q'_i(x)|}{q_i(x)} 
    + \E\big( |R'(x)| \mid R(x)=\z \big) \rho_{R(x)}(\z) ,
$$
which completes the proof. 
\end{proof}

We can finally provide the proof of the main result.

\begin{proof}[Proof of Theorem~\ref{th:main}] 
%We will apply Proposition~\ref{th:mainbound}.
The right-hand term in the statement of Proposition~\ref{th:mainbound}
can be bounded with Proposition~\ref{th:EBound}. 
Indeed, due to Lemma~\ref{le:d2bounds} we know that 
$\varpi_{k_1}(a) \ \le\  e\; |a|^{\frac{1}{2k_1}-1}$ for all~$a$.
Applying Proposition~\ref{th:EBound} with the parameters $B=1$, 
$C=e$, and $\d=(2k_1)^{-1}$ yields 
$$
 \E \big(|R'(x)| \mid R(x) = a \big) \rho_{R(x)}(a) \ \le\ 
  e(2k_1 + 1) \, \sum_{i=2}^m \frac{|w'_i(x)|}{\max\{ |w_i(x)|, |w_i(x)|^{1-\frac{1}{2k_1}}\} }.
$$
Applying Proposition~\ref{th:mainbound} implies 
for $x\in\R$ and $a\in\R^*$, recalling that $w_i(x) :=q_i(x) x^{d_i}$, 
\begin{eqnarray} \notag
%\begin{split}
 \E \big(|F'(x)| \mid F(x) = a \big) \rho_{F(x)}(a) \ &\le& \ 
 \frac{1}{\sqrt{2\pi}} \sum_{i=1}^m \sum_{j=1}^{k_i} \b'_{ij}(x) 
 +\sum_{i=1}^m \left| \frac{q'_i(x)}{q_i(x)} \right| \\ \label{eq:STAR}
&& \phantom{x} +\phantom{x} 
  e(2k_1 + 1) \, \sum_{i=2}^m \frac{|w'_i(x)|}{\max\{ |w_i(x)|, |w_i(x)|^{1-\frac{1}{2k_1}}\} } 
 =: g(x) .
%\end{split}
\end{eqnarray} 
$$
 \int_{0}^{1} \frac{|w_i'(x)|}{\max\{ w_i(x), w_i(x)^{1-\frac{1}{2k_1}}\}}\, dx \ \le\ 
  2 \LV(q_i) + k_it + 2k_1 .
$$
%xxx
The function $g(x)$ on the right-hand side of \eqref{eq:STAR} is integrable:
% namely, using Proposition~\ref{pro:new}, 
$$ 
 \int_0^1 g(x)\, dx \ \le\ 
  \frac{1}{\sqrt{2\pi}} \sum_{i=1}^m \sum_{j=1}^{k_i} (t-1)
  +  \sum_{i=1}^m \LV (q_i)
  + e(2k_1 + 1) \, \sum_{i=2}^m (2\LV(q_i) + k_i t + 2k_1) < \infty .
$$
By Proposition~\ref{pro:new} we can bound
$\LV(q_i) \le \frac12 2 k_i \ln t$.  
Moreover, Theorem~\ref{th:rice} 
can be applied (see Lemma~\ref{le:SingF})
and states that 
$\E(\#\{x\in [0,1] : F(x) = 0 \}) \le \int_0^1 g(x)\, dx$.
Hence, 
\begin{eqnarray}\label{eq:AF} \notag
%\begin{split}
  \E(\#\{x\in [0,1] : F(x) = 0 \}) 
 \ &\le\ & \frac{1}{\sqrt{2\pi}} (k_1+\ldots + k_m) (t-1)  +  (k_1+\ldots + k_m) \ln t \\ \notag
 && +\phantom{x} e(2k_1 + 1) \big(
 (k_2 + \ldots + k_m) (2\ln t + t)   + (m-1) 2k_1 \big)\\
 && = \Oh( k^2 m t) ,
%\end{split}
\end{eqnarray}
where $k$ denotes the maximum of the $k_i$. 

The number of zeros of $F$ in $[1,\infty)$ equals the number of zeros $x\in (0,1]$ of $F(x^{-1})$. 
Moreover, $F(x^{-1})$ has the same structure as $F$ except that the supports $S_{ij}$ 
are replaced by $-S_{ij}$. Since we can shift the degrees without changing the number of positive zeros, 
we conclude that 
$\E(\#\{x\in [1,\infty) : F(x) = 0 \})$ is also bounded by~\eqref{eq:AF}.
Therefore, $\E(\#\{x\in \R : F(x) = 0 \})$ is upper bounded by four times~\eqref{eq:AF}.
\end{proof}

%%%
%\bibliographystyle{plain}
%\bibliography{fewnom_refs}

\end{document}